\documentclass[11pt]{article}
\usepackage{amsmath,amsthm,amsfonts,amscd,eucal,latexsym,amssymb,color}
\usepackage{hyperref}
\usepackage{lmodern}
\usepackage{epsfig}  
\usepackage{color}
\definecolor{NicoColor}{RGB}{56, 174, 199}

\newtheorem{theorem}{Theorem}[section]
\newtheorem{lemma}[theorem]{Lemma}
\newtheorem{proposition}[theorem]{Proposition}
\newtheorem{corollary}[theorem]{Corollary}
\theoremstyle{definition}
\newtheorem{definition}[theorem]{Definition}

\theoremstyle{remark}
\newtheorem{remark}{Remark}

\def\bR{{\mathbb R}}

\newcommand{\supp}{{\rm supp}}
\newcommand{\myrightleftarrows}{\mathrel{\substack{\longrightarrow \\[-.6ex] \longleftarrow}}}

\begin{document}

\par 
\bigskip 
\par 
\rm 
 
\par 
\bigskip 
\LARGE 
\noindent 
{\bf Constructing Hadamard states via an extended M\o ller operator} 
\bigskip 
\par 
\rm 
\normalsize 

\large
\noindent {\bf Claudio Dappiaggi$^{1,2,a}$},
 {\bf Nicolo' Drago}$^{3,4,b}$ \\
\par
\small
\noindent $^1$ 
Dipartimento di Fisica, Universit\`a degli Studi di Pavia,
Via Bassi, 6, I-27100 Pavia, Italy.\smallskip

\noindent$^2$  Istituto Nazionale di Fisica Nucleare - Sezione di Pavia,
Via Bassi, 6, I-27100 Pavia, Italy.\smallskip

\noindent$^3$ Dipartimento di Matematica, Universit\`a di Genova - Via Dodecaneso 35, I-16146 Genova, Italy.\smallskip

\noindent$^4$ Istituto Nazionale di Fisica Nucleare - Sezione di Genova, Via Dodecaneso, 33 I-16146 Genova, Italy.

\bigskip

\noindent $^a$  claudio.dappiaggi@unipv.it,
$^b$  drago@dima.unige.it
 \normalsize

\par 
 
\rm\normalsize 
\bigskip
\noindent {\small Version of \today}

\rm\normalsize 
 
 
\par 
\bigskip 

\noindent 
\small 
{\bf Abstract}. 
We consider real scalar field theories whose dynamics is ruled by normally hyperbolic operators differing only by a smooth potential $V$. By means of an extension of the standard definition of M\o ller operator, we construct an  isomorphism between the associated spaces of smooth solutions and between the associated algebras of observables. On the one hand such isomorphism is non-canonical since it depends on the choice of a smooth time-dependant cut-off function. On the other hand, given any quasi-free Hadamard state for a theory with a given $V$, such isomorphism allows for the construction of another quasi-free Hadamard state for a different potential. The resulting state preserves also the invariance under the action of any isometry, whose associated Killing field commutes with the vector field built out of the normal vectors to a family of Cauchy surfaces, foliating the underlying manifold. Eventually we discuss a sufficient condition to remove on static spacetimes the dependence on the cut-off via a suitable adiabatic limit.

\vskip .3cm


\vskip .3cm


\normalsize
\bigskip 

\section{Introduction}
The algebraic approach is a mathematically rigorous scheme which is especially well-suited for formulating quantum theories also on globally hyperbolic spacetimes -- see \cite{Benini:2015bsa,Brunetti:2013maa} for recent reviews. In a few words it can be described as being based on two key objects. The first consists of the assignment to a physical system of a suitable $*$-algebra of observables, which encodes structural properties such as locality, causality and the dynamics. For free field theories this step has been thoroughly investigated and it is well-understood -- see for example \cite{Benini:2013fia,Benini:2015bsa,Fredenhagen:2014lda}. 

The second ingredient consists of the identification of an algebraic state, that is a positive, linear and normalized functional on the algebra of observables. Via the celebrated GNS theorem one can recover the usual probabilistic interpretation proper of quantum theories. Yet, between the plethora of possible states, not all are considered to be physically acceptable and a widely recognized criterion to single out the relevant ones is to require that the so-called Hadamard condition is fulfilled. Especially thanks to the work of Radzikowski \cite{Radzikowski:1996pa,Radzikowski:1996ei}, such condition has been translated in a constraint on the form of the wavefront set of the bidistribution associated to the two-point correlation function, built out of the state. The advantages of using the so-called {\em Hadamard states} are manifold, ranging from the finiteness of the quantum fluctuations of all observables to the possibility of constructing Wick polynomials following a covariant scheme -- see \cite{Khavkine:2014mta} for a recent exhaustive review. 

Once the Hadamard condition is assumed, several natural questions arise. The first concerns the existence of such states, a problem which was answered positively for free field theories (barring linearised gauge theories for which the method cannot be applied -- see for example  \cite{Benini:2014rya,Fredenhagen:2013cna,Gerard:2014jba,Wrochna:2014gia}) in \cite{fnw} by means of a spacetime deformation technique. 

The second is instead how to construct concretely Hadamard states. This problem has been investigated extensively in the last ten years and several options are available. If the underlying background is highly symmetric, such as the spacetimes of cosmological interest, one can resort to a mode expansion \cite{Olbermann:2007gn,Them:2013uka}, a technique which can be employed successfully to construct ground states of Hadamard form also if the metric is stationary \cite{Sahlmann:2000fh}. Other efficient techniques rely on pseudodifferential calculus \cite{Gerard:2012wb} or on the existence of a (conformal) null boundary of the underlying spacetime \cite{Dappiaggi:2005ci,Dappiaggi:2007mx,Dappiaggi:2008dk,Gerard:2014hla}. 

The third question, which is especially interesting from a physical point of view, is whether, in between the existing Hadamard states, there are some invariant under the action of all isometries of the underlying background. In general the answer to this question is negative, since there are well known cases, such as the massless, minimally coupled scalar field on globally hyperbolic, static spacetimes with compact Cauchy surfaces, when these states do not even exist. Even excluding these pathological scenarios, identifying isometry invariant, Hadamard states is not an easy task and a concrete way of doing it is known either in highly symmetric backgrounds, such as those of cosmological interest and those possessing a complete timelike Killing field, or for specific free fields with a special value of the mass and of the coupling of scalar curvature. An example is the massless and conformally coupled real scalar field on an asymptotically flat and globally hyperbolic spacetime for which Hadamard states, invariant under all isometries have been constructed in \cite{Dappiaggi:2005ci}.

Especially this last case leads to wonder whether, knowing a Hadamard, isometry invariant, algebraic state for a free field theory with given values of the parameters such as the mass and the coupling to scalar curvature, is not sufficient to build a counterpart with the same properties for other values of the parameters. The goal of this paper is to address this question. More precisely we shall tackle this problem for vector valued real scalar fields and a family of normally hyperbolic partial differential operators differing only for a smooth, formally self-adjoint potential $V$. Notice that this hypothesis includes in particular the case $V=m^2+\xi R$, $\xi\in\mathbb{R}$ and $R$ being the scalar curvature, which is in many cases the most interesting one. 
 We remark from the very beginning that we will address the whole problem from an abstract point of view, not entering into specific applications, such as for example constructing isometry invariant, Hadamard states for massive real scalar fields in asymptotically flat spacetimes, starting from the results of \cite{Dappiaggi:2005ci}. Such analysis would require a paper on its own and we plan to present it elsewhere. Furthermore we remark that, although we focus only on scalar fields, we expect that our procedure can be generalized almost straightforwardly to other cases, such as the Dirac and the Proca field. In the first case in particular the counterpart of our procedure might be related to other recent approaches \cite{Finster:2015wga}. 

The procedure, that we follow, elaborates on the initial stages of the approach followed in \cite{Dragon1} to prove the principle of perturbative agreement and most notably on the so-called M\o ller operator which is an intertwiner between the spaces of smooth solutions of free field theories whose dynamics is ruled by normally hyperbolic operators differing only for a smooth and compactly supported potential. Notice that, therefore, terms like a constant squared mass or a non minimal coupling to scalar curvature are excluded. 

As a first step in this paper we will show that, using carefully the domain of definition of the retarded and of the advanced Green operator of a normally hyperbolic PDO, the requirement on the support of the potential can be greatly relaxed, leading to what we call {\em extended M\o ller operator}. The main advantage of the extension is the possibility to combine it with the time-slice axiom in order to construct an explicit isomorphism between spaces of dynamical configurations for scalar fields whose dynamics differs by a term linear in the field, depending on a smooth external potential. Such isomorphism turns out to be non-canonical since it depends on the choice of a cut-off, function only of the time variable $t$. 

As a second step we will extend this new isomorphism to one between the $*$-algebras of fields of the two theories, showing as a by-product that it allows to pull-back states from one theory to the other. This procedure is remarkable since, if we start from a Hadamard state invariant under all background isometries, the pull-back preserves automatically the Hadamard condition and the invariance under all isometries whose associated Killing field commutes with the vector field built out of the normals to the Cauchy surfaces, individuating a foliation of the background, as specified in a theorem by Bernal and Sanchez -- see Section $2$. This result is already rather satisfactory since, in many concrete cases, only isometries of this kind are present. 

Yet, for the sake completeness and to check the robustness of our method, we investigate the possibility of removing the dependence on the cut-off function, a procedure which becomes especially important on stationary spacetimes when additional isometries are present. This limit procedure, which we will call {\em adiabatic limit}, can be well-defined only in a suitable weak sense. On the one hand we will show that, when we can take such limit, the resulting state will become invariant also under time translation. On the other hand, if we consider static, globally hyperbolic spacetimes, we can investigate the removal of the cut-off dependence at the level of modes. For the limit to be well-defined, we are able to write a sufficient condition, which is fulfilled on Minkowski spacetime. Furthermore, as a by-product, it turns out that, if one starts from a ground state, it ends up with a counterpart enjoying the same property.

The outline of the paper is as follows: In Section $2$ we introduce the extended M\o ller operator and we study its structural properties, showing in particular that it allows the identification of a non-canonical isomorphism between spaces of dynamical configurations for scalar fields whose dynamics differs by a term linear in the field, depending on a smooth, formally self-adjoint, external potential. In Section $3$, we extend such isomorphism at the level of the algebra of observables, while in Section $4$ we discuss the deformation argument at a level of algebraic states. In particular we discuss the Hadamard property, the interplay with the background isometries and the role of the cut-off function, which makes the above mentioned isomorphism non canonical. We discuss how such cut-off might be removed (adiabatic limit) and the additional properties of the algebraic state when such operation is well-defined. In the last Section, we discuss the adiabatic limit for static spacetimes by means of a mode expansion.

\section{The extended M\o ller operator - Classical configurations}\label{Section: The extended Moller operator - Classical configurations}

Consider an arbitrary finite dimensional vector bundle $E\equiv E[W,\pi,M]$ over a four-dimensional globally hyperbolic spacetime $M$.
Here $M$ comes endowed with a smooth Lorentzian metric $g$ of signature $(-,+,+,+)$ and its futher supposed to be oriented and time oriented.
We assume that $E$ is equipped with a fiberwise positive scalar product $\langle,\rangle_E$ and with a $\langle,\rangle_E$-compatible connection.
The latter one induces, together with the Levi-Civita connection, a connection $\nabla$ on $T^*M^{\otimes\ell}\otimes E$ for all $\ell\geq 0$.
Let $\Gamma(E)$, the space of smooth sections of $E$, be referred as the {\em space of kinematical configurations}.
In order to define a classical field theory we need to fix the dynamics, which we assume to be ruled by the following linear operator:
\begin{equation}\label{gen_dyn}
P_V\phi\doteq\left(P+V\right)\phi=0,\quad\phi\in\Gamma(E),
\end{equation}
where $P$ is a $\nabla$-compatible \cite[Lemma 1.5.5]{BGP}, normally hyperbolic operator while $V:\Gamma(E)\to\Gamma(E)$ is a linear term which, in every local coordinate neighbourhood $U$ of $M$, such that $E$ restricted thereon is trivial, is a smooth map from $U$ to $\textrm{End}(W)$. Observe that the choice of $\langle,\rangle_E$-compatible connection $\nabla$ has been made for later convenience, see Remark \ref{Remark: on induced bundles}, and since the prototype of a normally hyperbolic operator is $P\phi=\textrm{Tr}(\nabla^2\phi)$, $\phi\in\Gamma(E)$, namely the connection d'Alembert operator \cite[Example 1.5.2]{BGP} obtained through $\nabla$.
For all practical purposes, the main example that we have in mind consists of the case when $E=M\times\mathbb{R}$, $P=\square_g$ is the d'Alembert wave operator, while $V$ is a mass-like term such as $m^2+\xi R$, $\xi\in\mathbb{R}$, while $R$ is the scalar curvature of $(M,g)$.

Despite the fact that both $P,P_V$ are normally hyperbolic operators, the $V$-term in \eqref{gen_dyn} has been decoupled from $P$ since we want to consider it as a perturbation potential along the same lines followed by \cite{Dragon1} to prove the generalized principle of perturbative agreement. The key difference is that we do not require $V$ to be compactly supported.

Recall that, since $P_V$ is a normally hyperbolic operator, it possesses Green operators. Before defining them we introduce a class of notable functions:
\begin{definition}
We say that $f\in\Gamma(E)$ is {\bf past-compact} (-) or {\bf future compact} (+) if $\textrm{supp}(f)\cap J^\pm(p)$ is compact for all $p\in M$. We denote the space of these functions respectively as $\Gamma_{pc}(E)$ and $\Gamma_{fc}(E)$. We say that $f$ is {\bf timelike compact} if it is both past and future compact. We call $\Gamma_{tc}(E)\doteq\Gamma_{pc}(E)\cap\Gamma_{fc}(E)$.
\end{definition}

\noindent For every $\phi\in\Gamma(E)$ and for every $\alpha\in\Gamma_0(E)$, we introduce the pairing
\begin{equation}\label{pair1}
(\alpha,\phi)\doteq\int\limits_M d\mu_g \langle\alpha,\phi\rangle_E,
\end{equation}
where $\langle,\rangle_E$ is the scalar product of the fibres of $E$ and $d\mu_g$ is the metric-induced volume form. Henceforth we assume also that $V$ is formally self-adjoint with respect to the pairing \eqref{pair1}. Following \cite{Baernew,Sanders:2012ac},

\begin{definition}\label{Green}
We call {\em retarded Green operator}, $G^+_V:\Gamma_{pc}(E)\to\Gamma_{pc}(E)$ and {\em advanced Green operator}  $G^-_V:\Gamma_{fc}(E)\to\Gamma_{fc}(E)$ the unique maps such that 
$$P_V\circ G^\pm_V=G^\pm_V\circ P_V = id:\Gamma_{pc/fc}(E)\to\Gamma_{pc/fc}(E)$$
Additionally we call {\em causal propagator} $G_V\doteq G^+_V-G^-_V:\Gamma_{tc}(E)\to\Gamma(E)$. 
\end{definition}

Notice that the uniqueness of $G^\pm_V$ is a by-product of the results of \cite{Baernew,Sanders:2012ac}, once one notices that $P$ is a formally self-adjoint operator with respect to the pairing \eqref{pair1}. Additionally, since $P_V$ is normally hyperbolic, 
it holds that, for all $f\in\Gamma_{pc/fc}(E)$,
$$\textrm{supp}(G^\pm_V(f))\subseteq J^\pm(\textrm{supp}(f)).$$
The properties of the causal propagator can be summarized in the following short exact sequence of vector spaces:
\begin{equation*}
0\longrightarrow \Gamma_{tc}(E)\overset{P_V}{\longrightarrow} \Gamma_{tc}(E)\overset{G_V}{\longrightarrow} \Gamma(E)\overset{P_V}{\longrightarrow} \Gamma(E)\longrightarrow 0.
\end{equation*}
As by-product, if we call $\mathcal{S}_V(M)=\left\{\phi\in\Gamma(E)\;|\;P_V\phi=0\right\}$ the space of {\em dynamical configurations}, the following is an isomorphism of topological vector spaces:
\begin{equation}\label{iso1}
\mathcal{S}_V(M)\simeq\frac{\Gamma_{tc}(E)}{P_V\Gamma_{tc}(E)}.
\end{equation}
The isomorphism is realized via the causal propagator $G_V$ itself. 

We investigate whether it is possible to relate the spaces of dynamical configurations for different choices of the potential $V$. Notice that it is sufficient to answer to this question for the case when the dynamics is ruled on the one hand by $P$ and on the other hand by $P_V$ with any but fixed choice of $V\neq 0$. To this end, we generalize the so-called M\o ller operator -- see for example \cite{Dragon1}. As a first step we recall a notable property of globally hyperbolic spacetimes \cite{Bernal:2004gm,Geroch:1970uw}, here presented as in \cite[Th. 1.3.10]{BGP}:
\begin{proposition}\label{BS}
Let $(M,g)$ be a four dimensional, connected, orientable and time orientable spacetime. The following statements are equivalent:
\begin{enumerate}
\item $(M,g)$ is globally hyperbolic.
\item There exists a Cauchy surface in $(M,g)$. 
\item $(M,g)$ is isometric to $\mathbb{R}\times\Sigma$ endowed with the line element $ds^2=-\beta dt^2+h_t$ where $t:\mathbb{R}\times\Sigma\to\mathbb{R}$ is the projection on the first factor, $\beta$ is a smooth and strictly positive, scalar function on $\mathbb{R}\times\Sigma$ and $t\mapsto h_t$, $t\in\mathbb{R}$, yields a one-parameter family of smooth Riemannian metrics. Furthermore, for all $t\in\mathbb{R}$, $\{t\}\times\Sigma$ is a $3$-dimensional, spacelike, smooth Cauchy surface in $M$. 
\end{enumerate}
\end{proposition}

\noindent In order to simplify several statements, we fix once and for all the foliation of $(M,g)$ as $\bR\times\Sigma$ as well as a time coordinate $t$ as in Proposition \ref{BS}. Furthermore we introduce the following auxiliary structure,
\begin{definition}\label{Cauchy}
Let $(M,g)$ be a globally hyperbolic spacetime. We call a subset $\mathcal{N}$ of $M$ a {\bf Cauchy neighbourhood} if $(\mathcal{N},g|_{\mathcal{N}})$ is a globally hyperbolic submanifold of $M$ which contains a whole Cauchy surface of $M$. A smooth, past compact, function on $M$, depending only on the variable $t$, is called a {\em partition of unity subordinated to $\mathcal{N}$}, if it is equal to $1$ on $J^+(\mathcal{N})\setminus\mathcal{N}$ and to $0$ on $J^-(\mathcal{N})\setminus\mathcal{N}$.
\end{definition}

Henceforth we shall assume that a fixed Cauchy neighbourhood $\mathcal{N}$ has been taken and we introduce $\chi^+$, a partition of unity subordinated to it. At the same time we call $\chi^-(t)\doteq 1-\chi^+(t)\in C^\infty_{fc}(M)$. The following theorem is an important, albeit intermediate, step in our construction:
\begin{theorem}\label{DDtheorem}
Let  $\mathcal{S}_{V^+}(M)$ and $\mathcal{S}(M)$ be the space of dynamical configurations for \eqref{gen_dyn} with $V^+\doteq V\chi^+(t)\neq 0$ in the first case and for $V=0$ in the second one. We call {\bf extended M\o ller operator} the map 
\begin{equation}\label{DD}
R_{0,V^+}\doteq\mathbb{I}-G^+_{V^+}\,V^+ :\Gamma(E)\to\Gamma(E),
\end{equation}
where $G^+_{V^+}$ is the retarded Green operator of $P_{V^+}=P+V^+$ as in Definition \ref{Green}. $R_{0,V^+}$ is 
\begin{enumerate}
\item an automorphism of $\Gamma(E)$ whose inverse is $R^{-1}_{0,V^+}:\Gamma(E)\to\Gamma(E)$ such that
\begin{equation}\label{DD-1}
R_{0,V^+}^{-1}\varphi = \left(\mathbb{I}+G^+ V^+\right)\varphi,\quad\forall\varphi\in\Gamma(E)
\end{equation}
where $G^+$ is the retarded Green operator of $P$. 
\item an intertwiner between $P\equiv P_0$ and $P_{V^+}$, namely
$$P_{V^+}\circ R_{0,V^+}=P.$$
In other words $R_{0,V^+}$ implements a non-canonical isomorphism between $\mathcal{S}(M)$ and $\mathcal{S}_{V^+}(M)$.
\end{enumerate}
\end{theorem}

\begin{proof}
The first step consists of noticing that the extended M\o ller operator is well-defined. As a matter of fact, if $\phi\in\Gamma(E)$, it turns out that $\chi^+\phi\in\Gamma_{pc}(E)$, {\it i.e.} it lies in the domain of $G^+_{V^+}$. The smoothness of $R_{0,V^+}\phi$ is a 
by-product of the regularity properties of the advanced and retarded Green operators. To prove the first statement it suffices to show that $\mathbb{I}+G^+ V^+$ is both a left and a right inverse of $R_{0,V^+}$ on $\Gamma(E)$. Notice that the candidate inverse is well-defined on all elements of $\Gamma(E)$ for the same reason for which $R_{0,V^+}$, is. Hence on $\Gamma(E)$ it holds
\begin{align*}
R_{0,V^+}\circ R_{0,V^+}^{-1}= \left(\mathbb{I}-G^+_{V^+}\,V^+\right)\left(\mathbb{I}+G^+\,V^+\right)=\\
\mathbb{I}-G^+_{V^+}\,V^+ +G^+\,V^+ - G^+_{V^+}\,(P_{V^+}-P) G^+\,V^+ = \\
\mathbb{I}-G^+_{V^+}\,V^+ +G^+\,V^+ - G^+ V^+ + G^+_{V^+} V^+ = \mathbb{I},
\end{align*}
where we used both the identity $P_{V^+}-P=V^+$ and that $G^+V^+\varphi\in\Gamma_{pc}(E)$ for any $\varphi\in\Gamma(E)$. Therefore we are entitled to use the identities $P\circ G^+ = id|_{\Gamma_{pc}(E)}$ and $G^+_{V^+}\circ P_{V^+}=id|_{\Gamma_{pc}(E)}$. Notice that we have shown only that $R_{0,V^+}^{-1}$ is a right inverse of $R_{0,V^+}$, but the procedure to prove that it is also a left inverse is the same and, therefore, we shall omit it.

The second statement, namely $P_{V^+}\circ R_{0,V^+}=P$ descends from Definition \ref{Green} and from the following chain of identities:
\begin{equation}\label{aux1}
P_{V^+}\circ R_{0,V^+}=P+V^+ - (P_{V^+}\circ G^+_{V^+})\,V^+ =P+V^+-V^+=P.
\end{equation}
In other words, if $\phi\in\mathcal{S}(M)$, then $R_{0,V^+}\phi\in\mathcal{S}_{V^+}(M)$. Therefore the restriction of $R_{0,V^+}$ to $\mathcal{S}(M)$ identifies a non-canonical isomorphism between $\mathcal{S}(M)$ and $\mathcal{S}_{V^+}(M)$. Observe that \eqref{aux1} entails that $R_{0,V^+}$ is surjective since, for all $\varphi\in\mathcal{S}_{V^+}(M)$, $\phi\doteq R^{-1}_{0,V^+}\varphi\in\mathcal{S}(M)$ and $R_{0,V^+}\phi=\varphi$.
\end{proof}
With a slight abuse of notation, henceforth we shall indicate with $R_{0,V^+}$ the restriction of the extended M\o ller operator to $\mathcal{S}(M)$. The adjective {\em extended} comes from a comparison with the M\o ller operator used in \cite{Dragon1}. In this case it was required that the potential is compactly supported and smooth. 
\begin{remark}
Notice that the isomorphism implemented in \eqref{DD} is manifestly non-canonical since it depends explicitly on the choice of $\chi^+$, which, in turn, depends on the choice of $\mathcal{N}$. Also the choice of using the retarded Green operator is purely conventional. We could have considered the advanced Green operator with the caveat of replacing everywhere $\chi^+$ with $\chi^-\doteq 1-\chi^+$.
\end{remark}

In order to get rid of the localization in time of $V$, we will exploit a notable property of the space of smooth solutions to any free field theory, namely the time-slice axiom \cite{BGP}, which we restate in a slightly unconventional form.  

\begin{proposition}\label{Prop_timeslice}
Let $M$ be a globally hyperbolic spacetime and let the dynamics be ruled by \eqref{gen_dyn}. Choose a Cauchy neighbourhood $\mathcal{O}\subset M$ and $\eta^+\equiv\eta^+(t)$, a partition of unity subordinated to it, as per Definition \ref{Cauchy}.Then the {\bf time-slice operator}
\begin{equation}\label{timeslice}
TS^V:\mathcal{S}_V(M)\to\frac{\Gamma_{tc}(E)}{P_V[\Gamma_{tc}(E)]},\quad \varphi\mapsto TS^V(\varphi)=[P_V\eta^+\varphi],
\end{equation}
is an isomorphism of vector spaces. The inverse operator $TS^{V,-1}:\frac{\Gamma_{tc}(E)}{P[\Gamma_{tc}(E)]}\to\mathcal{S}_V(M)$ is  $[\alpha]\mapsto TS^{V,-1}[\alpha]\doteq G_V\alpha$, where $G_V$ is the causal propagator of $P_V$. 

Furthermore, let $\mathcal{S}_{V,sc}(M)$ be the collection of all smooth and spacelike compact solutions to \eqref{gen_dyn}, namely $\varphi\in\mathcal{S}_{V,sc}(M)$ if and only if $\supp(\varphi)\cap\Sigma$ is compact for any choice of a Cauchy surface of $M$. Then $TS^V|_{\mathcal{S}_{V,sc}(M)}$ is an isomorphism between $\mathcal{S}_{V,sc}(M)$ and $\frac{\Gamma_0(E)}{P[\Gamma_0(E)]}$.
\end{proposition} 

A proof of the properties of the time-slice operator has been given in \cite[Theorem 3.4]{Benini:2015bsa} for spacelike compact, smooth solutions to the Klein-Gordon equation. The extension to the case considered in the above proposition is a direct consequence of the properties of the domain of definition of the advanced and of the retarded Green operators.

\begin{remark}
Notice that fixing a Cauchy neighbourhood $\mathcal{O}$ in Proposition \ref{Prop_timeslice} implies actually that we can even consider the time-slice operator as an isomorphism between $\mathcal{S}_V(M)$ and $\frac{\Gamma_{tc}(E|_{\mathcal{O}})}{P_V[\Gamma_{tc}(E|_{\mathcal{O}})]}$, where $E|_{\mathcal{O}}$ is the restriction of the vector bundle $E$.
Different choices of $\eta^+$ with the same support properties, say $\eta^+_1$ and $\eta^+_2$ do not affect \eqref{timeslice} since, for every $\varphi\in\mathcal{S}_V(M)$, $(\eta^+_1-\eta^+_2)\varphi\in\Gamma_{tc}(E|_{\mathcal{O}})$. Hence $P_V((\eta^+_1-\eta^+_2)\varphi)$ is a representative of the trivial equivalence class in $\frac{\Gamma_{tc}(E|_{\mathcal{O}})}{P_V[\Gamma_{tc}(E|_{\mathcal{O}})]}$ .

In our analysis this perspective will become useful and, to distinguish it, we will indicate it via 
$$TS^V_{\mathcal{O}}:\mathcal{S}_V(M)\to\frac{\Gamma_{tc}(E|_{\mathcal{O}})}{P_V[\Gamma_{tc}(E|_{\mathcal{O}})]},\quad \varphi\mapsto TS^V(\varphi)=[P_V\eta^+\varphi].$$
Notice that $TS^V_{\mathcal{O}}$ is an isomorphism whose inverse $TS^{V,-1}_{\mathcal{O}}$ is still implemented by $G_V$, the causal propagator of $V$ acting on (the representatives of) the equivalence classes of $\frac{\Gamma_{tc}(E|_{\mathcal{O}})}{P_V[\Gamma_{tc}(E|_{\mathcal{O}})]}$. With a slight abuse of notation we will often indicate $G_V$ in place of $TS^{V,-1}_{\mathcal{O}}$.
\end{remark}

Assume that $\mathcal{O}\subset J^+(\mathcal{N})\setminus\mathcal{N}$, then for every $\alpha\in \Gamma_{tc}(E|_{\mathcal{O}})$, $P_{V^+} \alpha= P_V \alpha$ since $\chi^+$ is identically $1$ on $\mathcal{O}$. In other words $P_{V^+}[\Gamma_{tc}(E|_{\mathcal{O}})]$ is isomorphic to $P_V[\Gamma_{tc}(E|_{\mathcal{O}})]$. From this it descends the following statement

\begin{proposition}\label{alliso}
Let $M$ be a globally hyperbolic spacetime, $\mathcal{N},\mathcal{O}$ two Cauchy neighbourhoods of $M$. Let $\mathcal{O}\subset J^+(\mathcal{N})\setminus\mathcal{N}$ and let $\chi^+\equiv\chi^+(t)$ be a a partition of unity subordinated to $\mathcal{N}$. Then, the following is a chain of isomorphisms:
\begin{equation}\label{Isomain}
\mathcal{S}(M)\mathrel{\mathop{\myrightleftarrows}^{R_{0,V^+}}_{R^{-1}_{0,V^+}}}\mathcal{S}_{V^+}(M)\mathrel{\mathop{\myrightleftarrows}^{TS^{V^+}_{\mathcal{O}}}_{G_{V^+}}}\frac{\Gamma_{tc}(E|_{\mathcal{O}})}{P_{V^+}[\Gamma_{tc}(E|_{\mathcal{O}})]}\mathrel{\mathop{\myrightleftarrows}^{{\chi^+|_{\mathcal{O}}=1}}_{id|_{\mathcal{O}}}}\frac{\Gamma_{tc}(E|_{\mathcal{O}})}{P_V[\Gamma_{tc}(E|_{\mathcal{O}})]}
\mathrel{\mathop{\myrightleftarrows}^{G_V}_{TS^V_{\mathcal{O}}}}\mathcal{S}_V(M),
\end{equation}
where $\mathcal{S}_V(M)$ is the space of smooth solutions to \eqref{gen_dyn}, while $TS^{V^+}_\mathcal{O}$ is defined as in \eqref{timeslice}.
\end{proposition}

\begin{proof}
All isomorphisms have been already proven or they descend from the general theory presented for example in \cite{BGP}. The only necessary remark is that \eqref{gen_dyn} is a linear wave equation and hence we can apply to it Proposition \ref{Prop_timeslice}.
\end{proof}

We introduce a symbol for the isomorphism established in Proposition \ref{alliso}, that is 
\begin{eqnarray}\label{mainiso}
\mathcal{R}_{0,V^+}:\mathcal{S}(M)\to\mathcal{S}_V(M).\\
\mathcal{R}_{0,V^+}(\varphi)=\left(G_V\circ TS^{V^+}_\mathcal{O}\circ R_{0,V^+}\right)\varphi,\notag
\end{eqnarray}
where we have translated in formulas \eqref{Isomain}. Notice that we kept the dependence on $\chi$ explicit since it reminds us that this $1:1$ map is non-canonical.

\section{The extended M\o ller operator - Algebra of observables}

In this section we want to extend the isomorphism \eqref{mainiso} from the level of smooth dynamical configurations to that of the algebra of observables. The construction of the latter is based on a well-established procedure, which we will recapitulate succinctly. For concreteness we will consider the dynamics ruled by \eqref{gen_dyn}, but, at all steps, the reader is free to replace $V$ with either $V^+$ or even with $0$. We shall follow mainly \cite{Benini:2013fia}. 

As a starting point we define a family of functionals, labelled by smooth and compactly supported functions on $M$, that is, for all $f\in \Gamma_0(E)$,
$$F_f:\Gamma(E)\to\mathbb{R}\quad \varphi\mapsto F_f(\varphi)=\int\limits_M d\mu_g \langle f, \varphi\rangle_E,$$
where $d\mu_g$ is the metric-induced volume measure and $\langle,\rangle_E$ the fiberwise scalar product on $E$. The next step consists of considering the restriction of the above functionals from $\Gamma(E)$ to $\mathcal{S}_V(M)$, the dynamically allowed configurations. As a consequence we individuate redundant functionals. In other words $F_f(\varphi)=0$ for all $\varphi\in\mathcal{S}_V(M)$ if and only if there exists $h\in \Gamma_0(E)$ such that $f=P_V h$. Consequently we introduce the following class of functionals, called {\em linear classical observables}
\begin{equation}\label{clobs}
\frac{\Gamma_0(E)}{P_V[\Gamma_0(E)]}\ni[f]\mapsto F_{[f]}:\mathcal{S}_V(M)\to\bR,\quad\varphi\mapsto F_{[f]}(\varphi)\doteq F_f(\varphi).
\end{equation}

This collection forms a vector space $\mathcal{E}^{obs}_V(M)$ which separates all configurations since, as a consequence of Hahn-Banach theorem, $\Gamma_0(E)$ separates the points in $\Gamma(E)$ with respect to the pairing \eqref{pair1}. 

Additionally we notice that $\mathcal{E}^{obs}_V(M)$ comes endowed with a weakly non-degenerate symplectic form:
\begin{equation}\label{sympl}
\sigma_V:\mathcal{E}^{obs}_V(M)\times\mathcal{E}^{obs}_V(M)\to\mathbb{R}\quad \left(F_{[f]},F_{[f^\prime]}\right)\mapsto\sigma_V(F_{[f]},F_{[f^\prime]})=\int\limits_M d\mu_g \langle f,G_V(f^\prime)\rangle_E,
\end{equation}
where $G_V$ is the causal propagator associated to $P_V$. The compactness of the support of both $f$ and $f^\prime$ guarantees the finiteness of the integral, while the properties of $G_V$ entail both skew-symmetry and the independence from the choice of representatives.
Starting from these premises we introduce the following unital $*$-algebra
$$\mathcal{T}^{obs}_V(M)=\bigoplus_{n=0}^\infty\mathcal{E}^{obs}_V(M;\mathbb{C})^{\otimes n},$$
where $\mathcal{E}^{obs}_V(M;\mathbb{C})\doteq\mathcal{E}^{obs}_V(M)\otimes\mathbb{C}$ while $\mathcal{E}^{obs}_V(M;\mathbb{C})^{\otimes 0}\doteq\mathbb{C}$. The $*$-operation is the standard extension of complex conjugation to the tensor product, which is unambiguously and fully determined by the relation $(F_{[f]}\otimes F_{[g]})^* = F_{\overline{[g]}}\otimes F_{\overline{[f]}}$, for all pairs of $[f],[g]\in\frac{\Gamma_0(E)}{P_V[\Gamma_0(E)]}$. 
As a last step we encode the canonical commutation relations by taking the quotient with respect to the $*$-ideal $\mathcal{I}_V(M)$ generated by $F_{[f]}\otimes F_{[f^\prime]}-F_{[f^\prime]}\otimes F_{[f]}-i\sigma_V(F_{[f]},F_{[f^\prime]})\mathbb{I}$, where $\sigma_V$ is given in \eqref{sympl} and $\mathbb{I}$ is the unit in $\mathcal{T}^{obs}_V(M)$. The end-point is the {\em algebra of observables}
\begin{equation}\label{algebraobs}
\mathcal{A}^{obs}_V(M)\doteq\frac{\mathcal{T}^{obs}_V(M)}{\mathcal{I}_V(M)}.
\end{equation}
Notice that, since $\mathcal{E}^{obs}_V(M)$ is isomorphic to its labelling space $\frac{\Gamma_0(E)}{P_V[\Gamma_0(E)]}$, then \eqref{algebraobs} coincides with the standard algebra of fields, whose construction can be found for example in \cite{BGP, Benini:2013fia}. We kept the subscript $V$ explicit through all our discussion to emphasize its role in building all the relevant structures. Henceforth we will switch often between $\mathcal{E}^{obs}_V(M)$ and its labelling space without mentioning explicitly the underlying isomorphism.

We are now ready to make contact with the analysis of the previous section by using \eqref{mainiso} in \eqref{clobs}. The result is the following family of functionals 
$$F_{[f]}\circ\mathcal{R}_{0,V^+}:\mathcal{S}(M)\to\mathbb{R}\quad [f]\in\frac{\Gamma_0(E)}{P_V[\Gamma_0(E)]}.$$
Notice that, in view of Theorem \ref{DDtheorem} and of Proposition \ref{Isomain}, each $F_{[f]}\in\mathcal{E}^{obs}(M)$. In order to characterize the properties of the whole family we need first two ancillary lemmas, aimed at studying the formal dual of the time-slice operator and of the extended M\o ller operator:
\begin{lemma}\label{Lemma_dualTS}
Let $\widetilde{(,)}:\mathcal{S}_{V,sc}(M)\times\frac{\Gamma_{tc}(E)}{P_V[\Gamma_{tc}(E)]}\to\mathbb{R}$ be the pairing
$$\widetilde{(\psi,[\alpha])}=(\psi,\alpha)\doteq\int\limits_M d\mu_g\langle\psi,\alpha\rangle_E.$$
Let $TS^V$ be the time-slice operator as in \eqref{timeslice}.  Then the formal dual of $TS^V$ is the operator $TS^{V,*}:\mathcal{S}_{V,sc}(M)\to \frac{\Gamma_0(E)}{P_V[\Gamma_0(E)]}$ such that 
$$TS^{V,*}(\psi)=[P_V\rho^-\psi].$$
Here $\rho^-\equiv\rho^-(t)\in C^\infty_{pc}(M)$ is a partition of unity subordinated to $\mathcal{O}^\prime$, a Cauchy neighbourhood lying in $J^+(\mathcal{O})\setminus\{\mathcal{O}\}$. 
Furthermore $TS^{V,*}=-TS^V|_{\mathcal{S}_{V,sc}(M)}$ and thus it is an isomorphism whose inverse is $-G_V$.
\end{lemma}

\begin{proof}
Let $\psi\in\mathcal{S}_{V,sc}(M)$ and $\varphi\in\mathcal{S}_V(M)$ be arbitrary. The following chain of identities holds true:
\begin{gather*}
\widetilde{\left(\psi,TS^V(\varphi)\right)}=\left(\psi,P_V\eta^+\varphi\right)=\left(G^-_VP_V\rho^-\psi,P_V\eta^+\varphi\right)=\\
=\left(P_V\rho^-\psi,G^+_VP_V\eta^+\varphi\right)=\left(P_V\rho^-\psi,\eta^+\varphi\right)=\left(P_V\rho^-\psi,\varphi\right),
\end{gather*}
where, in the second identity, we have used that $\rho^-$ is equal to $1$ on $\mathcal{O}$ and that $\rho^-\psi$ is future compact. In the third equality, we have used that the formal adjoint of $G^-_V$ is $G^+_V$, while in the last equality, we have removed $\eta^+$ since it is equal to $1$ on $\mathcal{O}^\prime$. On account of the arbitrariness of both $\psi$ and $\varphi$ and of the non-degenerateness of the pairing, it holds that $TS^{V,*}(\psi)=P_V\rho^-\psi$. Since $P_V\psi=0$,  $P_V\rho^-\psi=-P_V\rho^+\psi$, which, per direct comparison with \eqref{timeslice} entails that $TS^{V,*}=-TS^V|_{\mathcal{S}_{V,sc}(M)}$. Furthermore, on account of this identity, we can translate to $TS^{V,*}$ the properties of the time-slice operator as proven in \cite{Benini:2015bsa}, in particular the independence from the choice of the Cauchy neighbourhood and of the partition of unity subordinated to it.
\end{proof}

\begin{remark}\label{Rem_TS}
Notice that, if we choose a fixed Cauchy neighbourhood $\mathcal{O}^\prime$, the dual time-slice operator can be read also as an isomorphism between $\mathcal{S}_{V,sc}(M)$ and $\frac{\Gamma_0(E|_{\mathcal{O}^\prime})}{P_V[\Gamma_0(E|_{\mathcal{O}^\prime})]}$. In this case we will write $TS^{V,*}_{\mathcal{O}^\prime}$. Notice, moreover, that $TS^{V,*}_{\mathcal{O}^\prime}=-TS^V_{\mathcal{O}^\prime}$ on $\mathcal{S}_{V,sc}(M)$. 
\end{remark}

\begin{lemma}\label{Lemma_dualDD}
Let $R_{0,V^+}:\mathcal{S}(M)\to\mathcal{S}_{V^+}(M)$ be the extended M\o ller operator as per \eqref{DD} referred to the potential $V^+\doteq V\chi^+$. Let $\widetilde{(,)}^\prime:\mathcal{S}_{V^+}(M)\times\frac{\Gamma_0(E)}{P_{V^+}\Gamma_0(E)}\to\mathbb{R}$ be the pairing such that, for all $\varphi\in\mathcal{S}_{V^+}(M)$ and for all $[\alpha]\in\frac{\Gamma_0(E)}{P_{V^+}\Gamma_0(E)}$, 
\begin{equation}
\label{paring-new}
\widetilde{(\varphi,[\alpha])}^ =(\varphi,\alpha)\doteq\int\limits_M d\mu_g \langle\varphi,\alpha\rangle_E,
\end{equation}
where $\langle,\rangle$ is the fiberwise scalar product on $E$. Then the {\bf formal dual} operator of $R_{0,V^+}$ is $R^*_{0,V^+}:\frac{\Gamma_0(E)}{P_{V^+}\Gamma_0(E)}\to\frac{\Gamma_0(E)}{P\Gamma_0(E)}$ such that
\begin{equation}\label{dualDD}
R^*_{0,V^+}[\alpha]\doteq\left[\left(\mathbb{I}-V^+G^-_{V^+}\right)\alpha\right],
\end{equation}
where $\alpha$ is any representative of $[\alpha]$ and $G^-_{V^+}$ is the unique advanced Green operator associated to $P+V^+$. Furthermore $R^*_{0,V^+}$ is an isomorphism of vector spaces whose inverse is $R^{*,-1}_{0,V^+}[\alpha^\prime]\doteq[(\mathbb{I}+V^+G^-)\alpha^\prime]$, where $G^-$ is the advanced Green operator of $P$ and $[\alpha^\prime]\in\frac{\Gamma_0(E)}{P\Gamma_0(E)}$.
\end{lemma}

\begin{proof}
Let $u\in\mathcal{S}(M)$ and $[\alpha]\in\frac{\Gamma_0(E)}{P_{V^+}\Gamma_0(E)}$ be arbitrary. It holds
\begin{gather*}
\widetilde{(R_{0,V^+}u,[\alpha])}=(u,\alpha)-(G^+_{V^+}V^+ u,\alpha).
\end{gather*}
The second term on the right hand side can be written as
\begin{eqnarray*}
(G^+_{V^+}V^+ u,\alpha) &=&
\int_M d\mu_g \langle(G^+_{V^+}V^+ u), \alpha\rangle_E\\ &=&
\int_M d\mu_g \langle(G^+_{V^+}V^+ u), P_{V^+}G^-_{V^+}\alpha\rangle_E\\ &=&
\int_M d\mu_g(x) \langle V^+ u, G^-_{V^+}\alpha\rangle_E =
\int_M d\mu_g\langle u, (V^+ G^-_{V^+})\alpha\rangle_E,
\end{eqnarray*}
where in the second equality we have used that $P_{V^+}\circ G^-_{V^+}=id|_{\Gamma_0(E)}$. In the third, we have exploited both that $P_{V^+}$ is formally self-adjoint and that $\supp(V^+ u)\cap\supp(G^-_{V^+}\alpha)$ is compact.  Furthermore, the whole chain of identities does not depend on the choice of representative in $[\alpha]$, since, for all $P_{V^+}h$, $h\in\Gamma_0(E)$,
\begin{gather*}
\int_M d\mu_g \langle R_{0,V^+} u, P_{V^+}h\rangle_E=
\int_M d\mu_g \langle P_{V^+}R_{0,V^+}u,h\rangle_E=
\int_M d\mu_g\langle Pu, h\rangle_E=0,
\end{gather*}
where we used the compactness of $h$ together with the intertwining property proven in Theorem \ref{DDtheorem}. By merging together all identities, we obtain:
\begin{gather*}
\widetilde{(R_{0,V^+}u,[\alpha])}=
\widetilde{(u,[(\mathbb{I}-V^+G^-_{V^+})(\alpha)])}=
\widetilde{(u,R^*_{0,V^+}[\alpha])},
\end{gather*}
from which the sought result descends, being the pairing non-degenerate. To prove that $R^*_{0,V^+}$ is an isomorphism of vector spaces, it suffices to show that $R^{*,-1}_{0,V^+}\doteq\mathbb{I}+V^+G^-$ is indeed the inverse operator. For every $[\alpha^\prime]\in\frac{\Gamma_0(E)}{P\Gamma_0(E)}$ it holds
\begin{eqnarray*}
\left(R^*_{0,V^+}\circ R^{*,-1}_{0,V^+}\right)[\alpha^\prime]&=&
\left[\left(\mathbb{I}-V^+G^-_{V^+}+V^+G^- - V^+G^-_{V^+}V^+G^-\right)\alpha^\prime\right]\\ &=&
\left[\left(\mathbb{I}-V^+G^-_{V^+}+V^+G^- - V^+G^-_{V^+}(P_{V^+}-P)G^-\right)\alpha^\prime\right]=
[\alpha^\prime],
\end{eqnarray*}
where, in the last identity we used both that $P\circ G^- =id|_{\Gamma_{0}(E)}$ and that $G^-_{V^+}\circ P_{V^+}=id|_{\Gamma_{fc}(E)}$.
Analogously one shows that, for every $[\alpha]\in\frac{\Gamma_0(E)}{P_{V^+}\Gamma_0(E)}$, $\left(R^{*,-1}_{0,V^+}\circ R^*_{0,V^+}\right)[\alpha]=[\alpha]$ and this concludes the proof.
\end{proof}

\begin{corollary}\label{Cor_Green}
Let $\widetilde{R}^*_{0,V^+}\doteq\mathbb{I}-V^+G^-_{V^+}:\Gamma_0(E)\to\Gamma_0(E)$. Then it holds that $\widetilde R^*_{0,V^+}\circ P_{V^+}= P$ as well as
\begin{subequations}
\begin{equation}
R_{0,V^+}\circ G^+ = G^+_{V^+}:\Gamma_{pc}(E)\to\Gamma_{pc}(E),
\end{equation}
\begin{equation}\label{12b}
R_{0,V^+}\circ G \circ \widetilde R^*_{0,V^+}= G_{V^+}:\Gamma_0(E)\to\mathcal{S}_{V^+}(M),
\end{equation}
\end{subequations}
where $G^\pm$ and $G^\pm_{V^+}$ are the advanced and retarded Green operators of $P$ and $P_{V^+}$ respectively.
\end{corollary}

\begin{proof}
The first statement descends from the following chain of identities: Let $(,)$ be the pairing \eqref{pair1}. Then, in view of the domain of defintion and of the support properties of the advanced and retarded Green operators, for any $f\in\Gamma_{tc}(E)$ it holds
$$R_{0,V^+}G^+(f)=
G^+_{V^+}P_{V^+}R_{0,V^+}G^+(f)=G^+_{V^+}(f),$$
where we used the intertwining property proven in Theorem \ref{DDtheorem}. The arbitrariness of $f$ entails the sought conclusion. The second statement recasts in this framework the results of \cite[Lemma 3.1]{Dragon1}. Since the proof is identical up to minor modifications, we omit it.
\end{proof}

\noindent We have gathered all ingredients to state the main result of this section
\begin{proposition}\label{iso_algebra}
Let $\mathcal{R}^*_{0,V^+}:\frac{\Gamma_0(E)}{P_V[\Gamma_0(E)]}\to\frac{\Gamma_0(E)}{P[\Gamma_0(E)]}$ be such that, for every $[f]\in \frac{\Gamma_0(E)}{P_V[\Gamma_0(E)]}$
\begin{equation}\label{dual-Moller}
\mathcal{R}^*_{0,V^+}[f]=-\left(R_{0,V^+}^*\circ TS_{\mathcal{O}}^{V^+,*}\circ G_V\right)f,
\end{equation}
where $R_{0,V^+}^*$ is defined in Lemma \ref{Lemma_dualDD}, $\mathcal{O}$ is a Cauchy neighbourhood of $M$, lying in $J^+(\mathcal{N})\setminus\mathcal{N}$ and $TS^{V^+,*}_{\mathcal{O}}$ is the dual time-slice operator defined in Lemma \ref{Lemma_dualTS} and Remark \ref{Rem_TS}. Then $\mathcal{R}^*_{0,V^+}$ 
\begin{enumerate}
\item is the dual operator to \eqref{mainiso} on $\frac{\Gamma_0(E)}{P_V[\Gamma_0(E)]}$ with respect to the pairing \eqref{paring-new},
\item realizes a symplectic isomorphism between $(\mathcal{E}^{obs}_V(M),\sigma_V)$ and $(\mathcal{E}^{obs}(M),\sigma)$ via the map assigning to any $F_{[f]}\in\mathcal{E}^{obs}_V(M)$, $F_{[\widetilde f]}\in\mathcal{E}^{obs}(M)$ where $[\widetilde f]=\mathcal{R}^*_{0,V^+}[f]$.
\end{enumerate}
\end{proposition}

\begin{proof}
Let us start from {\em 1.} The statement descends directly from the definition of $\mathcal{R}_{0,V^+}$ and of the pairing \eqref{paring-new}, by plugging in the results of Lemma \ref{Lemma_dualTS} and of Lemma \ref{Lemma_dualDD}.

Let us focus on {\em 2.} Notice that, for every $[f]\in\frac{\Gamma_0(E)}{P_V[\Gamma_0(E)]}$ and for every $u\in\mathcal{S}(M)$
$$F_{[f]}(\mathcal{R}_{0,V^+}(u))=(\mathcal{R}_{0,V^+}(u),f)=\left(u,(R_{0,V^+}^*\circ TS_{\mathcal{O}}^{V^+,*}\circ G_V)(f)\right)=F_{[\widetilde{f}]}(u).$$ 
To show that we have identified an isomorphism between $\mathcal{E}^{obs}_V(M)$ and $\mathcal{E}^{obs}(M)$, we recall that in Lemma \ref{Lemma_dualDD} and in Lemma \ref{Lemma_dualTS}, we have already shown that each of the operators used to define $\mathcal{R}^*_{0,V^+}$ admits an inverse. Putting them together it turns out that $\mathcal{R}^{*,-1}_{0,V^+}\doteq TS^V _{\mathcal{O}}\circ G_{V^+}\circ R^{*-1}_{0,V^+}$ is both a left and a right inverse of $\mathcal{R}^*_{0,V^+}$. Notice that it is necessary that $\mathcal{O}$ is chosen to lie in $J^+(\mathcal{N})\setminus\mathcal{N}$. 

To conclude we need to prove that the symplectic forms are preserved. Let $[f],[f^\prime]\in\frac{\Gamma_0(E)}{P[\Gamma_0(E)]}$ and choose two representatives $f,f^\prime\in\Gamma_0(\mathcal{O})$. Then it holds:
\begin{gather*}
\sigma(\mathcal{R}^*_{0,V^+}[f],\mathcal{R}^*_{0,V^+}[f^\prime])=\int\limits_M d\mu_g\langle\left(R_{0,V^+}^*\circ TS_{\mathcal{O}}^{V^+,*}\circ G_V\right)f\,,\, G\left(R_{0,V^+}^*\circ TS_{\mathcal{O}}^{V^+,*}\circ G_V\right)f^\prime\rangle_E=\\
\int\limits_M d\mu_g \langle G_V(f)\,,\,\left( TS_{\mathcal{O}}^{V^+}\circ R_{0,V^+}\circ G\circ R_{0,V^+}^*\circ TS_{\mathcal{O}}^{V^+,*}\circ G_V\right)f^\prime\rangle_E=\\
-\int\limits_M d\mu_g\langle G_V(f)\,,\,\left( TS_{\mathcal{O}}^{V^+}\circ G_{V^+}\right)f^\prime\rangle_E=-\int\limits_M d\mu_g\langle G_V(f)\,,\, f^\prime\rangle_E = \sigma_V([f],[f^\prime])
\end{gather*}
In the third line we used both that $TS_{\mathcal{O}^\prime}^{V}\circ G_V$ is the identity and that, in view of Lemma \ref{Lemma_dualTS} and of Remark \ref{Rem_TS} $TS_{\mathcal{O}}^{V^+,*}=-TS_{\mathcal{O}^\prime}^{V}$. Furthermore we have also used \eqref{12b} as well as that $TS_{\mathcal{O}}^{V^+}\circ G_{V^+}=id|_{\frac{\Gamma_0(E|_{\mathcal{O}})}{P_{V^+}[\Gamma_0(E|_{\mathcal{O}})]}}$.
\end{proof}

Proposition \ref{iso_algebra} establishes via the extended M\o ller operator an isomorphism between the building blocks of the algebra of observables with and without the formally self-adjoint potential $V$, realized by the operator $\mathcal{R}^*_{0,V^+}$. As a by-product such result extends first of all to an isomorphism between $\mathcal{E}^{obs}(M;\mathbb{C})^{\otimes n}$ and $\mathcal{E}^{obs}_V(M;\mathbb{C})^{\otimes n}$ for all $n\geq 0$. Secondly, since $\mathcal{R}_{0,V^+}$ is a symplectomorphism, it turns out that also the ideal encoding the canonical commutation relations $\mathcal{I}_V(M;\mathbb{C})$ is isomorphic to $\mathcal{I}(M;\mathbb{C})$. Indicating the extension of $\mathcal{R}_{0,V^+}$ to $\mathcal{T}^{obs}(M)$ as $\mathcal{R}^{alg}_{0,V^+}$ we have identified ultimately the following isomorphism 
\begin{equation}\label{iso_alg}
\mathcal{R}^{alg,*}_{0,V^+}:\mathcal{A}^{obs}_V(M)\to\mathcal{A}^{obs}(M),
\end{equation}
where the action is unambiguously defined on the generators as in Proposition \ref{iso_algebra} and, for notational simplicity, we kept the symbol $\mathcal{R}^{alg,*}_{0,V^+}$ also for the map between the quotient algebras.

\section{The deformation argument}

In this section we are interested in studying algebraic states and their interplay with the extended M\o ller operator. Recall that, given any unital $*$-algebra $\mathcal{A}$, a (algebraic) state is a linear functional $\omega:\mathcal{A}\to\mathbb{C}$ such that

$$\omega(\mathbb{I})=1,\qquad \omega(a^*a)\geq 0\quad\forall a\in\mathcal{A},$$

where $\mathbb{I}$ is the unit element in $\mathcal{A}$. Via the celebrated GNS theorem, one recovers the probabilistic interpretation proper of quantum theories, namely to the pair $(\mathcal{A},\omega)$ it is associated a unique (up to unitary transformation) triple $(\mathcal{D}_\omega,\pi_\omega,\Omega_\omega)$ where $\mathcal{D}_\omega$ is a dense subspace of an Hilbert space $\mathcal{H}_\omega$ and $\pi_\omega\in Hom(\mathcal{A},\mathcal{L}(\mathcal{D}_\omega))$ is compatible with the $*$-operation. Furthermore $\Omega_\omega\in\mathcal{D}_\omega$ is a unit-norm cyclic vector such that $\overline{\pi_\omega(\mathcal{D}_\omega)\Omega_\omega}=\mathcal{H}_\omega$. 

Let us now restrict our attention to the algebra of fields $\mathcal{A}^{obs}_V(M)$ of a free field theory whose dynamics obeys \eqref{gen_dyn}. In this case, between the plethora of algebraic states, only a few of them can be considered as physically acceptable, this class going under the name of {\em Hadamard states}. The criterion which singles them out is well-known in the literature and we will not dwell into the details, leaving an interested reader to the recent review \cite{Khavkine:2014mta} and references therein. For the sake of this paper we will focus our attention on the so-called quasi-free/Gaussian states, that is $\omega:\mathcal{A}^{obs}_V(M)\to\mathbb{C}$ is completely reconstructed from its two-point function 
$$\omega_2:\Gamma_0(E)\otimes\Gamma_0(E)\to\mathbb{C},$$
subject to the constraints of being {\em a)} continuous with respect to the test-section topology of $\Gamma_0(E)$, \cite[Def. 1.1.1]{BGP}, {\em b)}
 positive, {\it i.e.} $\omega_2(\overline{f}\otimes f)\geq 0$ for all $f\in\Gamma_0(E)$, {\em c)} compatible with the equations of motion, that is $(P_V\otimes\mathbb{I})\omega_2=(\mathbb{I}\otimes P_V)\omega_2=0$ and {\em d)} compatible with the canonical commutation relations, namely $\omega_2(f\otimes f^\prime)-\omega_2(f^\prime\otimes f)=i\sigma_V(f,f^\prime)\mathbb{I}$, for all $f,f^\prime\in\Gamma_0(E)$, $\sigma_V$ being the symplectic form \eqref{sympl} and $\mathbb{I}$ being the identity element of the algebra. Notice that the two-point function identifies a bi-distribution which we indicate still as $\omega_2\in\mathcal{D}^\prime(E\times E)$. Here $V\in\Gamma(E)$ is arbitrary, possibly even vanishing.

Additionally we shall require $\omega_2$ to fulfill the so-called {\em Hadamard condition} which is a constraint on the singular structure of the two-point function expressed in terms of the associated wavefront set, see {\it e.g.}, \cite{Khavkine:2014mta}:
\begin{equation}\label{WF}
WF(\omega_2)=\{(x,k_x,y,k_y)\in T^*(M\times M)\setminus\{0\}\,|\,(x,k_x)\sim (y,-k_y)\,\textrm{and}\,k_x\triangleright 0\},
\end{equation}
where $\sim$ means that $x$ and $y$ are connected by a lightlike geodesic whose co-tangent vector at $x$ is $k_x$, while $-k_y$ is the parallel transport to $y$ of $k_x$ along such geodesic.
The symbol $\triangleright$ entails that $k_x$ is future directed. Notice that, for vector valued distributions, the wavefront set is defined as the union of that of all components \cite{sv}. From a physical point of view Hadamard states are of interest since on the one hand they reproduce the ultraviolet behaviour of the Poincar\'e vacuum and they guarantee that all quantum fluctuations of observables are finite. On the other hand they are the key ingredient in the definition of Wick polynomials, the building block for dealing with interactions in a perturbative scheme. Notice that, in this section, we could avoid the requirement of $\omega$ being quasi-free, checking instead the wavefront set of the truncated two-point function, hence exploiting the results of \cite{Sanders:2009sw}.

We will not be interested in this aspect, rather on the existence of Hadamard states. It is an old milestone that, for free field theories on a globally hyperbolic spacetime, one can guarantee via a {\em spacetime deformation argument} that a Hadamard state does exist \cite{fnw}, barring those exceptions mentioned in the introduction. The intrinsic drawback of the deformation procedure is that one loses any control on the invariance of the resulting state under the action of the background isometries. Ultimately this is an undesirable feature since, although not necessary from a technical point of view nor often explicitly spelled out, such invariance is certainly sought from a physical point of view. Unless impossible to construct ({\it e.g} a state invariant under time translation for a massless, minimally coupled, scalar field on a static, globally hyperbolic spacetime with compact Cauchy surfaces), in concrete models all computations are done with respect to a state invariant under the action of all isometries. For this reason, the first step is always to look for one of such Hadamard states.
In the past years several methods have been devised to identify one or more of these states, see for example \cite{Dappiaggi:2010gt,Dappiaggi:2007mx,Dappiaggi:2008dk,Olbermann:2007gn,Sanders:2013vza,Them:2013uka} and also \cite{Gerard:2016sgm,Vasy:2015jzp}.
However, unless one considers highly symmetric backgrounds, such as those homogeneous and isotropic, when a mode expansion is a viable tool, a construction scheme has been found only for specific field theories with a given value of the mass and of the coupling of scalar curvature.

The reason lies often in the possibility to exploit additional, special properties of the underlying equation of motion, such as, for example, a good behaviour of the space of dynamical configurations and of observables under a conformal rescaling of the metric -- see for example the case of a massless, conformally coupled scalar field on an asymptotically flat spacetime \cite{Dappiaggi:2005ci}.

We will show that by means of the extended M\o ller operator and of the results of the previous sections, one can also introduce a \textbf{\textit{ deformation argument in parameter space}}. In a few heuristic words, such argument guarantees that, if we can construct a Hadamard state, invariant under all background isometries, for a free field theory with a given value of the mass and of the coupling to scalar curvature, then one can induce a counterpart state for any value of these quantities. Furthermore such state fulfils the Hadamard condition and it is invariant under all background isometries, whose associated Killing field is at each point tangent to one of the Cauchy surfaces $\{t\}\times\Sigma$ of the underlying globally hyperbolic spacetime -- see Proposition \ref{BS}. One should be aware that our procedure deals with a more general scenario, since we will be able to induce Hadamard states from one free scalar field theory to another, provided that their dynamics differs only by a smooth potential. For practical purposes the most interesting cases are those in which such potential is of the form $m^2+\xi R$, $m$ being the mass and $R$ the scalar curvature, while $\xi\in\mathbb{R}$.

\begin{remark}\label{Remark: on induced bundles}
Whenever the manifold $M$ of interest is equipped with a complete Killing vector field $X$, we denote with $\alpha_X(\lambda):M\to M$ the action of the one-parameter subgroup of the isometry group, built from $X$ via the exponential map. Here $\lambda$ ranges over the whole $\mathbb{R}$ on account of the hypothesis of completeness of $X$.
The one-parameter subgroup acts on $C^\infty(M,W)$ by defining $(\alpha^\star_X(\lambda)f)(x)\doteq f(\alpha_X^{-1}(\lambda)x)$ for all $f\in C^\infty(M,W)$.
Similarly, we may lift the action of $\alpha_X(\lambda)$ on $\Gamma(E)$ by exploiting the choice of the $\langle,\rangle_E$-connection on $E$ made at the beginning of Section \ref{Section: The extended Moller operator - Classical configurations}.
Indeed, any curve $\alpha_X(\lambda)$ admits a unique parallel lift $\widehat{\alpha}_X(\lambda)$ on $E$ \cite[Chap. 2]{KN63}: we can then define the one-parameter subgroup $\widetilde{\alpha}_X(\lambda)\colon\Gamma(E)\to\Gamma(E)$ as $(\widetilde{\alpha}_X(\lambda)s)(x):=\widehat{\alpha}_X(\lambda)s(\alpha^{-1}_X(\lambda)(x))$, which is well defined section due to the lift relation $\pi\circ\widehat{\alpha}_X(\lambda)=\alpha_X(\lambda)\circ\pi$.
\end{remark}

For definiteness we consider smooth sections of a vector bundle, whose dynamics is ruled by a normally hyperbolic operator $P_V=P+V$ as described in Section \ref{Section: The extended Moller operator - Classical configurations}.
We start from the assumption that we have built a quasifree Hadamard state invariant under all background isometries when the dynamics is ruled by \eqref{gen_dyn} with $V=0$. Implicitly we are also assuming that the operator $P$ intertwines with the action of any isometry lifted to the smooth sections of the underlying vector bundle. The canonical example, that we have in mind, is the Klein-Gordon operator. Our main result, whose proof is inspired by the deformation argument of \cite{Fulling:1989nb}, is the following:

\begin{theorem}\label{main}
Let $(M,g)$ be a globally hyperbolic spacetime, isometric to $\bR\times\Sigma$ and $t:\bR\times\Sigma\to\bR$ a fixed time coordinate, compatible with Proposition \ref{BS}. Let $\xi$ be any complete Killing vector field of $\bR\times\Sigma$ such that $\mathcal{L}_\xi(\chi^+)=0$, where $\chi^+$ identifies a partition of unity subordinated to the Cauchy neighbourhood $\mathcal{N}$, as in Definition \ref{Cauchy}.  Let $\phi\in\Gamma(E)$ be such that $P\phi=0$. Let $\mathcal{A}^{obs}(M)$ be the associated algebra of observables \eqref{algebraobs} together with $\omega:\mathcal{A}^{obs}(M)\to\mathbb{C}$, a quasi-free Hadamard state. If $P_V=P+V$ is such that $V$ is a smooth potential and if $\mathcal{R}^{alg,*}_{0,V^+}:\mathcal{A}^{obs}_V(M)\to\mathcal{A}^{obs}(M)$ is the isomorphism \eqref{iso_alg}, then 
\begin{gather*}
\omega_V\doteq \omega\circ\mathcal{R}^{alg,*}_{0,V^+}:\mathcal{A}^{obs}_V(M)\to\mathbb{C},\\
a \mapsto \omega_V(a)\doteq\omega(\mathcal{R}^{alg,*}_{0,V^+}(a))
\end{gather*}
has the following properties:
\begin{enumerate}
\item It is a quasi-free Hadamard state,
\item Leaving implicit the isometry between $M$ and $\bR\times\Sigma$, it is invariant under the action of any $\xi$ as per hypothesis, provided that $\omega$ enjoys the same property and that $\mathcal{L}_\xi V=0$.
\end{enumerate}
\end{theorem}

\begin{proof}
Let $\omega_V$ be as per hypothesis. Since it is defined composing $\omega$ with $\mathcal{R}^{alg, *}_{0,V^+}$, which is in turn built out of a symplectic isomorphism between $\mathcal{E}^{obs}(M)$ and $\mathcal{E}^{obs}_V(M)$ intertwining the equations of motion, it inherits the property of being a quasi-free state. In order to check whether $\omega_V$ is Hadamard let us focus on its two-point function, namely, $\forall F_{[f]},F_{[f^\prime]}\in\mathcal{E}^{obs}_V(M)$,
$$\omega_V(F_{[f]}\otimes F_{[f^\prime]})=\omega(F_{\mathcal{R}^*_{0,V^+}[f]},F_{\mathcal{R}^*_{0,V^+}[f^\prime]})=\omega_2(R_{0,V^+}^*\circ TS_{\mathcal{O}}^{V^+,*}\circ G_V(f),R_{0,V^+}^*\circ TS_{\mathcal{O}}^{V^+,*}\circ G_V(f^\prime)),$$
where $\omega_2\in\mathcal{D}^\prime(E\times E)$ is the bi-distribution associated to $\omega$. Notice that, if we consider the restriction of $\omega_2$ to $\mathcal{O}$ the two-point function thereon is $\omega_{2,V^+}\doteq\omega_2(\widetilde R_{0,V^+}^*\bullet,\widetilde R_{0,V^+}^*\bullet)$ where $\widetilde R_{0,V^+}^*\doteq \mathbb{I}-V^+G^-_{V^+}$. The latter can also be read as the restriction to $\mathcal{O}$ of the two-point function of a quasi-free state for the sections of a vector bundle with dynamics built out of the time-dependant potential $V^+(t)=V\chi^+(t)$, $\chi^+$ being the partition of unity subordinated to $\mathcal{N}$. Notice that, since $\omega$ is of Hadamard form per hypothesis, we can apply \cite[Proposition 16]{Khavkine:2014mta} to conclude that also $\omega_{2,V^+}$ has such form since it coincides with $\omega_2$ in the past of $\mathcal{N}$ where $\chi^+$ vanishes. In turn since the two-point function of $\omega_V$ coincides with $\omega_{2,V^+}$ on $\mathcal{O}$, $\omega_V$ is a Hadamard state.

Let us focus on {\em 2.} and let us call as $\alpha_\xi(\lambda):M\to M$ the action of the one-parameter subgroup of the isometry group, built from $\xi$ via the exponential map. Notice that $\lambda$ ranges over the whole $\mathbb{R}$ on account of the hypothesis of completeness of $\xi$. In the sense of Remark \ref{Remark: on induced bundles} $\widetilde{\alpha}_X(\lambda)$ induces an action on $\Gamma(E)$. Assume that $\omega$ is a quasi-free Hadamard state for $\mathcal{A}^{obs}(M)$ invariant under the action of all isometries. Furthermore, in view of the hypothesis $\mathcal{L}_\xi(\chi^+)=0$, we notice that $R^*_{0,V^+}\widetilde\alpha_\xi(\lambda)=\widetilde\alpha_\xi(\lambda)R^*_{0,V^+}$, for all $\lambda\in\mathbb{R}$, the only non-trivial identity being $G^-_{V^+}\circ\widetilde\alpha_\xi(\lambda)=\widetilde\alpha_\xi(\lambda)\circ G^-_{V^+}$ on $\Gamma_0(E)$. The latter descends from the following chain of identities: For all $f\in\Gamma_0(E)$ and $\alpha\in\Gamma_{tc}(E)$
\begin{gather*}
((G^-_{V^+}\circ\widetilde\alpha_\xi(\lambda))f,\alpha)=(f,(\widetilde\alpha_\xi(-\lambda)\circ G^+_{V^+})\alpha)=(P_{V^+}G^-_{V^+}f,(\widetilde\alpha_\xi(-\lambda)\circ G^+_{V^+})\alpha)\\
=(G^-_{V^+}f,(P_{V^+}\circ\widetilde\alpha_\xi(-\lambda)\circ G^+_{V^+})\alpha)=(G^-_{V^+}f,(\widetilde\alpha_\xi(-\lambda)\circ P_{V^+}\circ G^+_{V^+})\alpha)=\\
=(G^-_{V^+}f,\widetilde\alpha_\xi(-\lambda)\alpha)=(\widetilde\alpha_\xi(\lambda)G^-_{V^+}f,\alpha),
\end{gather*}
where $(,)$ is the non-degenerate pairing between $\Gamma_0(E)$ and $\Gamma_{tc}(E)$. Notice that we have implicitly used the following property: Since $\xi$ is a Killing field, then $\supp(\widetilde{\alpha}_\xi(\lambda)\phi)\subseteq\alpha_\xi(\lambda)(\supp(\phi))$ for all $\phi\in\Gamma(E)$. The same result applies if we restrict our attention to sections lying in $\Gamma_0(E)$ or in $\Gamma_{pc/fc}(E)$. We used that the pairing is the one induced from the metric, that the inverse of $\widetilde\alpha_\xi(\lambda)$ is $\widetilde\alpha_\xi(-\lambda)$ and that $P_{V^+}$ is built out of the metric and of $\chi^+(t)$. To summarize, for every $[f],[f^\prime]\in\frac{\Gamma_0(E)}{P_{V}(\Gamma_0(E))}$, once chosen two representatives $f,f^\prime\in\Gamma_0(E|_{\mathcal{O}})$ where $\mathcal{O}$ is a Cauchy neighbourhood such that $\chi^+|_{\mathcal{O}}=1$, it holds that 
\begin{gather*}
\omega_{2,V}(\widetilde\alpha_\xi(\lambda)(f),\widetilde\alpha_\xi(\lambda)(f^\prime))=\omega(\widetilde R_{0,V^+}^*(\widetilde\alpha_\xi(\lambda)(f)),\widetilde R_{0,V^+}^*(\widetilde\alpha_\xi(\lambda)(f^\prime)))=\\
=\omega((\widetilde\alpha_\xi(\lambda)\circ \widetilde R_{0,V^+}^*(f)),(\widetilde\alpha_\xi(\lambda)\circ \widetilde R_{0,V^+}^*(f^\prime))=\omega(\widetilde R_{0,V^+}^*(f), \widetilde R_{0,V^+}^*(f^\prime))=\omega_{2,V}(f,f^\prime).
\end{gather*}
This implies invariance of $\omega_V$ under the action of the isometry built out of $\xi\in\Gamma(TM)$.
\end{proof}

The main drawback of the previous theorem is the lack of any control on the action on $\omega_{2,V}$ of any isometry whose associated Killing field $\xi^\prime$ is such that $\mathcal{L}_{\xi^\prime}\chi^+\neq 0$. While in many interesting scenarios, such as cosmological or Bianchi spacetimes, there are none of such symmetries, from a structural point of view, this is certainly an unwanted feature. As a concrete example consider any static, globally hyperbolic, four-dimensional spacetime with a non-compact Cauchy surface and consider the ground or a KMS state $\omega$ on $\mathcal{A}^{obs}(M)$ for a massless, minimally coupled scalar field, built out of the timelike Killing field. It enjoys the Hadamard property as proven in \cite{Sahlmann:2000fh}. Yet, the corresponding state on $\mathcal{A}^{obs}_V(M)$ built as per Theorem \ref{main}, while still Hadamard fails to be either a ground or a KMS state. While one might argue that, in these scenarios, one can proceed to a direct construction of Hadamard states without resorting to Theorem \ref{main}, the question, whether the dependence on the cut-off can be removed under suitable conditions, is a worthy one. 

A priori we expect that there is no positive answer in all possible scenarios, since it is known that certain free field theories, {\it e.g.} the massless and minimally coupled scalar field on four-dimensional de Sitter spacetime, do not possess a ground state even though their massive counterpart does. 

We seek for a way of implementing a sort of adiabatic limit with respect to the cut-off function $\chi^+$. A direct approach is doomed to failure since, if we let the cut-off function tend to the constant function $1$, the dual to the improved M\o ller operator $R^*_{0,V^+}$ converges to $R^*_{1,V}\doteq\mathbb{I}-VG^-_V$ in the topology of $\Gamma(E)$. In other words the image of a compactly supported, smooth section is no longer compact and, hence, there is no guarantee that $\omega_2(R^*_{1,V}\cdot,R^*_{1,V}\cdot)$ is a well-defined bi-distribution. Therefore we follow a strategy very similar in spirit to the one used in \cite{Fredenhagen:2013cna} to cope with an almost identical problem in defining thermal states for interacting field theories at a perturbative level. 


More precisely let $\omega$ be a quasi-free Hadamard state for $\mathcal{A}^{obs}(M)$ and let $\omega_V$ be the counterpart for $\mathcal{A}^{obs}_V(M)$ built as per Theorem \ref{main}. 
Let $\omega_{2,V}\in\mathcal{D}^\prime(E\times E)$ be its two-point function. For any $\lambda>0$ let $V^+_\lambda\doteq V\chi^+_\lambda$, where $\chi^+_\lambda(t)\doteq\chi^+(\frac t\lambda)$. For later convenience and without loss of generality $\chi^+$ is chosen as in Definition \ref{Cauchy} with the additional requirement that $\chi^+$ is a partition of unity subordinated to an arbitrary, but fixed Cauchy neighborhood $\mathcal{N}$ such that the Cauchy surface at $t=0$ lies in $J^+(\mathcal{N})\setminus\mathcal{N}$. Notice that this entails $\chi^+(0)=1$ and, consequently, $\chi^+_\lambda(0)=1$ for all $\lambda >0$. In addition choose a second Cauchy neighborhood $\mathcal{O}\subset J^+(\mathcal{N})\setminus\mathcal{N}$ and, for every $[f]\in\frac{\Gamma_0(E)}{P_{V^+_\lambda}\Gamma_0(E)}$, let $R_{0,V^+,\lambda}^*[f]=[\widetilde R_{0,V^+,\lambda}(f)]\doteq[(\mathbb{I}-V^+_\lambda G^-_{V^+_\lambda})(f)]$. Then, for any $f,f^\prime\in\Gamma_0(E)$, we call
\begin{gather}
\mathcal{R}_{0,V^+,\lambda}^*\doteq R_{0,V^+,\lambda}^*\circ TS_{\mathcal{O}}^{V^+_\lambda,*}\circ G_V,\label{lambda_R}\\
\omega_{2,V,\lambda}(f,f^\prime)\doteq\omega_2(\mathcal{R}_{0,V^+,\lambda}^*[f]\otimes\mathcal{R}_{0,V^+,\lambda}^*[f^\prime]),\label{lambda_state}
\end{gather}
where $[f],[f^\prime]$ are the equivalence classes in $\frac{\Gamma_0(E)}{P\Gamma_0(E)}$ built out $f$ and $f^\prime$ respectively.
The adiabatic limit consists of taking the limit for $\lambda\to\infty$ of this last expression. Two questions arise, namely if the limit exists and, if so, which are the properties of the resulting state. In the next proposition we address this second issue, leaving a partial answer to the first one to the next section.
\begin{proposition}\label{limit_properties}
Let $\omega:\mathcal{A}^{obs}(M)\to\mathbb{C}$ be a quasi-free Hadamard state and let $f,f^\prime\in\Gamma_0(E)$. Define $\omega_{2,V,\lambda}$ as \eqref{lambda_state} and let 
\begin{gather*}
\widetilde{\omega}_{2,V}(f,f^\prime)\doteq\lim\limits_{\lambda\to\infty}\omega_{2,V,\lambda}(f,f^\prime).
\end{gather*}
If the limit exists, it defines a distribution $\widetilde{\omega}_{2,V}\in\mathcal{D}^\prime(E\times E)$, such that it identifies a quasi-free state $\widetilde{\omega}_V$ for $\mathcal{A}_V^{obs}(M)$. Such a state is invariant under the action both of time-translation if this is an isometry of $(M,g)$ and of any complete Killing field $\xi\in\Gamma(TM)$ which, up to the isometry between $M$ and $\bR\times\Sigma$ satisfies the hypotheses of Theorem \ref{main}, provided also that $\omega$ enjoys the same property and that $\mathcal{L}_nV=\mathcal{L}_\xi V=0$, where $n$ is the vector field which coincides at each point with $\frac{\partial}{\partial t}$, $t$ being chosen as in the third item of Proposition \ref{BS}.
\end{proposition}

\begin{proof}
The statement is a direct consequence of the assumption that $\omega$ is a quasi-free state. With respect to the canonical commutation relations (CCRs), it is enough to observe that, for any $f,f'\in\Gamma_0(E)$ it holds $(G^\pm_{V^+,\lambda}f,f')=(G^\pm_Vf,f')$ for sufficiently large $\lambda$.
As a consequence $G_{V^+,\lambda}$ converges weakly to $G_V$, thus $\widetilde{\omega}_{2,V}$ is compatible with the CCRs.
On account of Theorem \ref{main}, each $\omega_{2,V,\lambda}$ is a quasi-free Hadamard state. If the limit as $\lambda\to\infty$ is well-defined, $\widetilde{\omega}_{2,V}$ preserves the property of being positive. Furthermore, for any $f,f^\prime\in\Gamma_0(E)$, notice that there exists a finite value of $\overline{\lambda}\in\mathbb{R}$ such that $P_Vf=P_{V^+_{\overline{\lambda}}}f$, where $V^+_{\overline{\lambda}}= V\chi^+_{\overline{\lambda}}$. In particular this entails that, for all $\lambda\geq\overline{\lambda}$, $\omega_{2,V,\lambda}(P_Vf, f^\prime)=0$ and, consequently, once taking the limit $\lambda\to\infty$, $\widetilde{\omega}_{2,V}(P_Vf,f^\prime)=0$. Analogously the same result holds true for $\widetilde{\omega}_{2,V}(f,P_Vf^\prime)$ with $f$ and $f^\prime\in\Gamma_0(E)$ arbitrary. Hence $\widetilde{\omega}_{2,V}$ specifies a unique quasi-free state on $\mathcal{A}^{obs}_V(M)$. 

Now let us consider any complete Killing field $\xi\in\Gamma(TM)$ and let $\alpha_\xi(\eta):M\to M$, $\eta\in\mathbb{R}$ be the action on $M$ of the associated $1$-parameter group of isometries, while $\widetilde{\alpha}_\xi(\eta):\Gamma(E)\to\Gamma(E)$ is the action induced by $\alpha$ on the smooth sections of $E$. We observe that $G^-_{V^+_\lambda}\circ\widetilde{\alpha}_\xi(\eta)=\widetilde{\alpha}_\xi(\eta)G^-_{V^+_{\lambda,\eta}}$, where $V_{\lambda,\eta}\doteq V\chi^+_{\lambda,\eta}$, being $\chi^+_{\lambda,\eta}\doteq(\widetilde{\alpha}_\xi(-\eta)\chi^+)_{\lambda}$. To prove this identity, first of all we compute, for any value of $\lambda$ and of $\eta$:
\begin{eqnarray*}
P_{V^+_{\lambda,\eta}}\left(\widetilde{\alpha}_\xi(-\eta)\circ G^-_{V^+_\lambda}\circ\widetilde{\alpha}_\xi(\eta)\right)&=&
\widetilde{\alpha}_\xi(-\eta)\circ P\circ G^-_{V^+_\lambda}\circ\widetilde{\alpha}_\xi(\eta)+
V^+_{\lambda,\eta}\circ\widetilde{\alpha}_\xi(-\eta)\circ G^-_{V^+_\lambda}\circ\widetilde{\alpha}_\xi(\eta)\\ &=&
\widetilde{\alpha}_\xi(-\eta)\circ P_{V^+_\lambda}\circ G^-_{V^+_\lambda}\circ\widetilde{\alpha}_\xi(\eta)=id|_{\Gamma_{tc}(E)}.
\end{eqnarray*}
In the second equality, we used that $P$ commutes per assumption with the action of all isometries and that $\widetilde{\alpha}_\xi(\eta)^{-1}=\widetilde{\alpha}_\xi(-\eta)$.
A similar calculation proves that $ \left(\widetilde{\alpha}_\xi(-\eta)\circ G^-_{V^+_\lambda}\circ\widetilde{\alpha}_\xi(\eta)\right)P_{V^+_{\lambda,\eta}}=id|_{\Gamma_{tc}(E)}$. Additionally, for every $\varphi\in\Gamma_{tc}(E)$, it holds:
\begin{gather*}
\supp\left(\left(\widetilde{\alpha}_\xi(-\eta)\circ G^-_{V^+_\lambda}\circ\widetilde{\alpha}_\xi(\eta)\right)(\varphi)\right)=
\alpha_\xi(-\eta)\left(\supp\left(\left( G^-_{V^+_\lambda}\circ\widetilde{\alpha}_\xi(\eta)\right)(\varphi)\right)\right)\subseteq\\
\subseteq\alpha_\xi(-\eta) J^-(\supp(\widetilde{\alpha}_\xi(\eta)(\varphi)))=J^-(\supp(\varphi)).
\end{gather*}
In other words $\widetilde{\alpha}_\xi(-\eta)\circ G^-_{V^+_\lambda}\circ\widetilde{\alpha}_\xi(\eta)$ is an advanced Green operator for $P_{V^+_{\lambda,\eta}}$ on $\Gamma_{tc}(E)$. Since $P_{V^+_{\lambda,\eta}}$ is formally self-adjoint with respect to the standard pairing between $\Gamma_{tc}(E)$ and $\Gamma_0(E)$, the results of \cite{Baernew} guarantee us that the advanced and retarded fundamental solutions are unique. Hence $\widetilde{\alpha}_\xi(-\eta)\circ G^-_{V^+_\lambda}\circ\widetilde{\alpha}_\xi(\eta)$ coincides with $G^-_{V^+_{\lambda,\eta}}$. As a consequence it holds that, for any fixed value of $\lambda$ and for any $\eta\in\mathbb{R}$,
$$\widetilde R^*_{0,V^+,\lambda}\circ\widetilde{\alpha}_\xi(\eta)=\widetilde{\alpha}_\xi(\eta)\circ \widetilde R^*_{\widetilde{\alpha}_\xi(\eta)\left(0,V^+,\lambda\right)},$$
where $\widetilde R^*_{\widetilde{\alpha}_\xi(\eta)\left(0,V^+,\lambda\right)}\doteq \mathbb{I}-V^+_{\lambda,\eta}\circ G^-_{V^+_{\lambda,\eta}}$. Combining this information with the property of the time-slice operator $TS_{\mathcal{O}}^{V^+_\lambda, *}\circ \widetilde{\alpha}_\xi(\eta)=\widetilde{\alpha}_\xi(\eta)\circ TS_{\alpha_\xi(-\eta)(\mathcal{O})}^{\widetilde{\alpha}_\xi(-\eta)(V^+_\lambda), *}$, it holds that 
\begin{equation}\label{R-commute}
\mathcal{R}^*_{0,V^+,\lambda}\circ\widetilde{\alpha}_\xi(\eta)=\widetilde{\alpha}_\xi(\eta)\circ\mathcal{R}^*_{\widetilde{\alpha}_\xi(\eta)\left(0,V^+,\lambda\right)},
\end{equation}
where $\mathcal{R}^*_{\widetilde{\alpha}_\xi(\eta)\left(0,V^+,\lambda\right)}\doteq R^*_{\widetilde{\alpha}_\xi(\eta)\left(0,V^+,\lambda\right)}\circ TS_{\alpha_\xi(-\eta)(\mathcal{O})}^{\widetilde{\alpha}_\xi(-\eta)(V^+_\lambda), *}\circ G_V$. To conclude, consider any pair $f,f^\prime\in\Gamma_0(E)$ and suppose that $\xi$ is one of the Killing fields as per hypothesis; it holds
$$\omega_{2,V,\lambda}(\widetilde{\alpha}_\xi(\eta)(f),\widetilde{\alpha}_\xi(\eta)(f^\prime))=
\omega_2(\mathcal{R}^*_{\widetilde{\alpha}_\xi(\eta)\left(0,V^+,\lambda\right)}[f], \mathcal{R}^*_{\widetilde{\alpha}_\xi(\eta)\left(0,V^+,\lambda\right)}[f^\prime]),$$
where we used the hypothesis of invariance of $\omega$. The case when $\xi$ is a Killing field such that $\mathcal{L}_\xi\chi^+=0$ and $\mathcal{L}_\xi V=0$ entails that $\mathcal{R}^*_{\widetilde{\alpha}_\xi(\eta)\left(0,V^+,\lambda\right)}=\mathcal{R}^*_{\left(0,V^+,\lambda\right)}$. Hence, as already proven in Theorem \ref{main}, for each $\lambda\in\mathbb{R}$, the corresponding state is invariant under the action of these isometries and the same holds true necessarily when taking the limit as $\lambda\to\infty$. When $\xi$ coincides with $n$, the effect of $\widetilde{\alpha}_n$ and of $\alpha_n$ is only to translate in the time-direction. Hence, taking the limit for $\lambda\to\infty$, under the additional hypothesis that $\mathcal{L}_n V=0$, the left hand side converges to $\widetilde\omega_{2,V}(\widetilde{\alpha}_n(\eta)(f),\widetilde{\alpha}_n(\eta)(f^\prime))$ while the right hand side to $\widetilde\omega_{2,V}(f,f^\prime)$ which is the sought statement.
\end{proof}

\begin{remark}
Notice that, in the above proposition, we have no control of the action of those isometries whose associated complete Killing field $\xi$ is neither proportional to $n$ nor fulfilling the hypotheses of Theorem \ref{main}. Hence, when present, invariance under the action of these remaining isometries should be checked case by case. The prime example consists of applying this method to Minkowski spacetime starting from the vacuum state for a free massive or massless scalar field. In this case, the adiabatic limit exists and it coincides with the ground state. Hence it is invariant under the action of the whole Poincar\'e group.
\end{remark}

It remains to be shown that there exist cases in which the adiabatic limit can be taken and thus the hypotheses of Proposition \ref{limit_properties} are met.

\section{The adiabatic limit on static spacetimes -- sufficient conditions}

As a last step, we discuss the convergence of \eqref{lambda_state} on static spacetimes, showing in particular that it can be reduced to studying the properties of the eigenfunctions of the Laplace-Beltrami operator built out of the metric induced on a Cauchy surface. This procedure generalizes the strategy used in \cite{Dragon1} to deal with the adiabatic limit on Minkowski spacetime. For the rest of the section we will make the following additional assumptions:
\begin{enumerate}
	\item $(M,g)$ is a globally hyperbolic, ultrastatic spacetime with a non compact Cauchy surface. With reference to Proposition $\ref{BS}$, this entails that $\beta=1$ and that, for each $t\in\mathbb{R}$, the metric $h_t$ is time-independent.
	\item As starting point, we consider a minimally coupled real scalar field $\phi\in C^\infty(M)$ with a fixed value $m_1$ of the mass. The equation of motion is $P\phi=\left(-\partial^2_t -K-m^2_1\right)\phi=0$, where $K$ is minus the Laplace
	 operator built out of $h$, the Riemannian metric on any Cauchy surface of $M$. The role of the potential $V$ is played by $m^2_2-m^2_1$, where $m_2$ is another, possibly vanishing, value of the mass. 
	\item We consider $\omega:\mathcal{A}^{obs}(M)\to\mathbb{C}$ to be a quasi-free ground state, which is of Hadamard form in view of the results of \cite{Sahlmann:2000fh}. 
\end{enumerate}

As the title of the section suggests, the whole construction works also for static spacetime. The price to pay is that the function $\beta$ becomes dependent also on the spatial coordinates. Hence the Klein-Gordon operator becomes of the form considered in item {\em 2.} only if one replaces $\phi$ with $\phi^\prime=\beta^{-2}\phi$ and $K$, though still elliptic, is no longer just the Laplace operator, but acquires extra terms. We feel that avoiding these additional complications enhances the clarity of the presentation.
 The main tool that we will be making use of is a mode expansion which for static, globally hyperbolic spacetimes, has been discussed in \cite{Fulling:1989nb} and, from a more complete mathematical point of view, recently in \cite[Section 2]{Avetisyan}. We refer especially to this last paper for more technical details. Following \cite{Fulling:1989nb}, we decompose the two-point function of a ground state as
\begin{equation}\label{mode-expansion}
\omega_{2,j}(x,y)=\int d\mu(k) T_{k,j}(t)\overline{T}_{k,j}(t^\prime)\psi_k(\underline{x})\overline{\psi}_{k}(\underline{y}),
\end{equation}
where $x=(t,\underline{x})$, $y=(t^\prime,\underline{y})$ while $j=1,2$ is introduced to distinguish between the values $m_1$ and $m_2$ of the mass, which are our starting and arrival point respectively. Here each $\psi_k$ is an eigenfunction of the positive elliptic operator $K$, that is $K\psi_k=\lambda^2_k\psi_k$, $\lambda_k\in\mathbb{R}$ and $d\mu(k)$ is a suitable spectral measure -- see \cite{Avetisyan}. At the same time each $T_{k,j}(t)$ obeys to the ODE:
\begin{equation}\label{modes}
\ddot{T}_{k,j}(t) + (\lambda^2_k+m^2_j)T_{k,j}(t)=0,
\end{equation}
where $j=1,2$ distinguishes between the case when we consider the mass $m_1$ or $m_2$. 
Furthermore the following normalisation condition is imposed for the Wronski determinant:
$$\overline{\dot{T}_{k,j}(t)}T_{k,j}(t)-\overline{T_{k,j}(t)}\dot{T}_{k,j}(t)=i.\quad\forall k,\,\textrm{and}\,j=1,2$$
We select the solution $T_{k,1}(t)=\frac{e^{-i\omega_{k,1}t}}{\sqrt{2\omega_{k,1}}}$
 (positive frequency splitting) with $\omega^2_{k,1}=\lambda^2_k+m^2_1$. A similar expansion can be written for the two-point function $\widetilde\omega_{2,\chi}(\cdot,\cdot)\doteq\omega_{2,1}(\widetilde R^*_{0,V^+}\cdot,\widetilde R^*_{0,V^+})$ where $\widetilde R^*_{0,V^+}=\mathbb{I}-V^+G^-_{V^+}$ with $V^+=(m^2_2-m^2_1)\chi^+(t)$ and where $G^-_{V^+}$ is the advanced fundamental solution of $P_{V^+}=P+V^+$. In this case $\omega_{2,\chi}$ acquires the same form of \eqref{mode-expansion}, though with $T_{1,k}(t)$ and $\overline{T}_{1,k}(t)$ replaced by $T_{k,\chi}(t)$ and its complex conjugate, fulfilling the ordinary differential equation
\begin{equation}\label{midmodes}
\ddot{T}_{k,\chi}(t) + (\lambda^2_k+m^2_1 + (m^2_2-m^2_1)\chi(t))T_{k,\chi}(t)=0,
\end{equation} 
with the initial condition $T_{k,\chi}(t_0)=T_{1,k}(t_0)$ for $t_0\in M\setminus\supp(\chi^+)$.
Notice that, in comparison to \eqref{lambda_state}, we are discarding the contribution of $TS^{V^+_\lambda,*}_{\mathcal{O}}\circ G_V$. This is legitimate since, in view of Lemma \ref{Lemma_dualTS} and of the properties of $\mathcal{O}$, the difference between any $f\in C^\infty_0(M)$ and its image under the action of these operators lies in the image of $P_V$, which vanishes when evaluated on a state for $\mathcal{A}^{obs}_V(M)$.

In order to perform the adiabatic limit $\chi\to 1$, first of all we write the cut-off $\chi(t)$ in the form $\int\limits_{-\infty}^tds\,f(s)$ with $f\in C^\infty_0(\mathbb{R})$ and such that $\int_{-\infty}^0ds\,f(s)=1$. If we replace such $\chi$ with $\chi_n(t)\doteq\chi(\frac{t}{n})$, the limit $\chi\to 1$ corresponds to $n\to\infty$.

Furthermore we observe that formula (\ref{mode-expansion}) expresses the two-point function of the state as the application of $f_j^{t,t'}$ on $K$, being $f_j^{t,t'}(\lambda_k)\doteq T_{k,j}(t)\overline{T_{k,j}(t')}$. Actually for any $u,v\in C^\infty_0(M)$ it holds
\begin{gather}\label{2 point function with spectral calculus}
\omega_{2,j}(u,v)=\int\limits_{\bR\times\bR} dtdt' \left(f_j^{t,t'}\right)(K)\left(u(t,\cdot)v(t',\cdot)\right)\qquad j=1,2.
\end{gather}
In a similar way $\widetilde{\omega}_{2,n}\doteq{\widetilde{\omega}_{2,\chi_n}}$ can be expressed as the application of $f_n^{t,t'}$ on $K$, where now $f_n^{t,t'}(\lambda_k)\doteq T_{k,n}(t)\overline{T_{k,n}(t')}$, being $T_{k,n}$ the modes obtained from $T_{k,\chi}$ by substituting $\chi\to\chi_n$.
Therefore the problem of addressing the adiabatic limit is reduced to show that the sequence $(f_n^{t,t'})(K)$ converges to $(f_2^{t,t'})(K)$ in an operational meaningful way\footnote{Note that the time integration will definitely play no role since all integrals  in (\ref{2 point function with spectral calculus}) are compactly supported.}.

Since $K$ is self adjoint on $L^2(\Sigma)$, where $\Sigma$ is an arbitrary but fixed Cauchy surface, and $f_2^{t,t'},f_n^{t,t'}$ are all measurable functions on its spectrum, standard results of spectral analysis allow us to conclude that the wanted convergence is reached whenever the following conditions are met:
\begin{enumerate}
\item $f_n^{t,t'},f_2^{t,t'}$ are bounded on $\sigma(K)$;
\item $\sup_n\|f_n^{t,t'}\|_\infty<\infty$.
\item $f_n^{t,t'}$ converges pointwise to $f_2^{t,t'}$;
\end{enumerate}

The first two conditions are met whenever we consider $m_1,m_2>0$.
In the case when $m_1=0$, $\widetilde{\omega}_{k,n}(t)\doteq \lambda^2_k+m^2_1 + (m^2_2-m^2_1)\chi_n(t)$ vanishes if the eigenvalue $\lambda_k=0$.
This specific case is the one responsible for the failure of this procedure when we try to apply it to a massless, minimally coupled, real scalar field on a globally hyperbolic, static spacetime with compact Cauchy surfaces.
Yet, in the general case we can still give a sufficient condition, which involves further properties of the spectral measure $d\mu(k)$, to ensure the convergence. Rearranging the eigenvalues so that $\lambda_k$ vanishes if $k=0$, the following lemma holds true:

\begin{lemma}
In the above hypotheses, if $m_1,m_2>0$, then the two-point function $\widetilde\omega_{2,n}$ converges to a two-point function for a real scalar field with mass $m_2$.
If $m_1=0$ ($m_2=0$) and the spectral measure $d\mu(k)$ is such that
\begin{gather}\label{condition on the spectral measure}
\lim_{\epsilon\to 0^+}\int_{|\lambda_k|<\epsilon}\frac{1}{|\lambda_k|}d\mu(k)=0,
\end{gather}
then the sequence of two-point functions $\widetilde\omega_{2,n}$ is convergent. Furthermore the limit two-point function identifies a ground state, which is of Hadamard form.
\end{lemma}
\begin{proof}
The idea of the proof is a natural generalization of the one present in \cite{Dragon1} (see also \cite[Section 2]{Avetisyan}) for the case of Minkowski spacetime.

First we deal with the case $m_1,m_2>0$.
In order to control the limit of $f_n^{t,t'}$ as $n\to\infty$ for all values of the time coordinate $t$, we introduce auxiliary modes, namely
\begin{gather}
T_{k,n,a}(t)\doteq\frac{1}{\sqrt{2\omega_{k,n}(t)}}\exp\left(-i\int\limits_{t_0}^t ds\,\omega_{k,n}(s)\right),
\end{gather}
where $t_0$ is an (inessential) arbitrary but fixed times such that $t_0\in J^-\supp(\chi^-)$.
Notice that, when $n\to\infty$, $T_{k,n,a}(t)\overline{T_{k,n,a}(t')}\to f_2^{t,t'}(k)$ i.e. the product of the auxiliary modes pointwisely converges to the product of those necessary to construct the ground state for a real scalar field of mass $m_2$, cfr \eqref{modes}. As a next step we prove that, in the same limit, then $\left|T_{k,n,a}-T_{k,n}\right|$ tends to $0$. We observe that $T_{k,n,a}(t)$ obeys to the differential equation
\begin{equation}\label{modes2}
\ddot{T}_{k,n,a}+\left(\omega^2_{k,n}+\delta_{k,n}\right)T_{k,n,a}=0,
\end{equation}
where 
$$\delta_{k,n}(t)=\frac{1}{2}\frac{\ddot{\omega_{k,n}}}{\omega_{k,n}}-\frac{3}{4}\frac{\dot{\omega}_{k,n}^2}{\omega^2_{k,n}}.$$
If we regard \eqref{modes2} as \eqref{midmodes} plus a potential $-\delta^2_{k,n}(t)$ which can be treated as a perturbation, we can construct $T_{k,n}$ via a Dyson series as
\begin{gather*}
T_{k,n}(t)=T_{k,n,a}(t)+
\sum\limits_{\ell=1}^\infty(-1)^\ell\int\limits_{-\infty}^t dt_1\int\limits_{-\infty}^{t_1} dt_2...\int\limits_{-\infty}^{t_{\ell-1}} dt_\ell\;\widehat{\Delta}_R^0(t,t_1)\widehat{\Delta}_R^0(t_1,t_2)\dots\widehat{\Delta}_R^0(t_{\ell-1},t_\ell)\notag\\
\delta_{k,n}(t_1)...\delta_{k,n}(t_\ell)T_{k,n,a}(t_\ell),
\end{gather*} 
where $\Delta_R^0$ is the retarded propagator of the unperturbed equation, namely, for any $h\in C^\infty_0(\mathbb{R})$, 
$$\widehat{\Delta}_R^0(h)(t)=\int\limits_{-\infty}^td\tau\,\frac{\sin(\int\limits_\tau^t\omega_{k,n}(s)ds)}{\sqrt{\omega_{k,n}(\tau)\omega_{k,n}(t)}}h(\tau).$$
Similarly to \cite{Dragon1}, the following estimate descends from the Dyson series:
\begin{gather}
|T_{k,n}-T_{k,n,a}|\leq\frac{1}{2\sqrt{\omega_{k,n}}}\left|-1+\exp(\int\limits_{\mathbb{R}}d\tau\,\frac{|\delta_{k,n}(\tau)|}{\omega_{k,n}(\tau)})\right|\leq\notag\\
\leq \frac{C}{2\sqrt{\omega_{k,n}}}\left(\int\limits_{\mathbb{R}}d\tau\,|(m^2_2-m^2_1)\frac{d^2(\omega^2_{k,n})}{d\tau^2}|+\int\limits_{\mathbb{R}}d\tau\,\left(\frac{d(\omega^2_{k,n})}{d\tau}\right)^2\right),\label{estimate}
\end{gather}
where $C$ is a suitable positive constant.
It then follows from \eqref{estimate} that $|T_{k,n}-T_{a,n,k}|$ tends to $0$, proving the pointwise convergence $f_n^t\to f_2^t$.

It remains to show that $\supp_n\|f_n^t\|_\infty<\infty$: in view of the definition of the modes $T_{k,n,a}(t)$, we notice that 
$$|T_{k,n}(t)|\leq|T_{k,n}(t)-T_{k,n,a}(t)|+|T_{k,n,a}(t)-T_{2,k}(t)|+|T_{2,k}(t)|\leq \frac{C^\prime}{\sqrt{2\widetilde{\omega}_k}},$$
where $C^\prime$ is a positive constant, $T_{2,k}$ are the modes built out of \eqref{modes}, replacing $m_1$ with $m_2$ and where the last inequality descends from \eqref{estimate}, having defined $\widetilde{\omega}_k=\min(\omega_{1,k},\omega_{2,k})$. 

This concludes the proof for the case $m_1,m_2>0$.
We stress that the above arguments relies on the assumption that we started from a massive real scalar field, since, in the massless case, $\omega_{k,n}$ might not be bounded from below by a positive constant.

Assume now for simplicity $m_1=0,\ m_2>0$: we introduce the energy for the modified modes
\begin{gather}\label{energy for the retarded modes}
E_{k,n}(t)\doteq|\dot{T}_{k,n}(t)|^2+\omega_{k,n}^2(t)|T_{k,n}(t)|^2.
\end{gather}
By choosing $\chi$ to be monotonically increasing we get that $\omega_{k,n}^2$ is monotonically increasing too, and therefore
\begin{gather*}
\frac{d}{dt}\left(\frac{E_{k,n}}{\omega_{k,n}^2}\right)=
-|T_{k,n}|^2\omega_{k,n}^{-4}\frac{d}{dt}\omega^2_{k,n}\leq 0.
\end{gather*}
From that we can infer that
\begin{gather*}\label{estimate for modes in the massless case}
|T_{k,n}(t)|^2\leq
\frac{E_{k,n}(t)}{\omega_{k,n}^2(t)}\leq
\frac{E_{k,n}(t_0)}{\omega_{k,n}^2(t_0)}=\frac{1}{|\lambda_k|},
\end{gather*}
having chosen $t_0$ such that $\chi_n(t_0)=0$.

Thanks to property (\ref{condition on the spectral measure}), we can now write the difference $\omega_{2,n}(f,f^\prime)-\widetilde\omega_{2}(f,f^\prime)$, where $f,f^\prime\in C^\infty_0(M)$, as an integral in the measure $d\mu(k)$, splitting it into the contribution for $|\lambda_k|\gtrless\epsilon$ respectively.

The first contribution tends to zero as $n\to\infty$ on account of the same arguments used for the massive case, while the latter one can be estimated, see (\ref{estimate for modes in the massless case}), as $\int_{|\lambda_k|<\epsilon}|\lambda_k|^{-1}d\mu(k)$ which tends to zero as $\epsilon\to 0$.

To conclude the proof, it suffices to notice that \eqref{mode-expansion} itself guarantees that we have built a ground state for a real scalar field with mass $m_2$. This is of Hadamard form as a consequence of the analysis in \cite{Sahlmann:2000fh}.
\end{proof}

\section*{Acknowledgements}
The authors are grateful to Christian G\'erard, Felix Finster, Thomas-Paul Hack and Nicola Pinamonti for the useful discussions. The work of C.D. is supported by the University of Pavia and he is grateful to the department of mathematics of the University of Genoa for the kind hospitality during the realization of this work. The work of N.D. is supported by a Ph.D. grant of the university of Genoa and he is grateful to the department of physics of the University of Pavia and to the Laboratoire de Math\'ematiques of the Universit\'e Paris Sud for the kind hospitality during the realization of this work.

\end{document}